\documentclass[reqno]{amsart}%
\usepackage{amsfonts}
\usepackage{verbatim}
\usepackage{xcolor}
\usepackage{amsmath}
\usepackage{amssymb}
\usepackage{graphicx}
\usepackage{accents}%
\providecommand{\U}[1]{\protect\rule{.1in}{.1in}}
\newtheorem{theorem}{Theorem}[section]

\newtheorem{corollary}[theorem]{Corollary}

\newtheorem{proposition}[theorem]{Proposition}
\newtheorem{remark}[theorem]{Remark}
\newtheorem{lemma}[theorem]{Lemma}

\numberwithin{equation}{section}
\begin{document}
\title[KdV equation]{Trace formulas revisited and a new representation of KdV solutions with
short-range initial data}
\author{Alexei Rybkin}
\address{Department of Mathematics and Statistics, University of Alaska Fairbanks, PO
Box 756660, Fairbanks, AK 99775}
\email{arybkin@alaska.edu}
\thanks{The author is supported in part by the NSF grant DMS-2009980.}
\date{August, 2024}

\begin{abstract}
We put forward a new approach to Deift-Trubowitz type trace formulas for the
1D Schrodinger operator with potentials that are summable with the first
moment (short-range potentials). We prove that these formulas are preserved
under the KdV flow whereas the class of short-range potentials is not.
Finally, we show that our formulas are well-suited to study the dispersive
smoothing effect.

\end{abstract}
\maketitle
\dedicatory{We dedicate this paper to Vladimir Marchenko on the occasion of his centennial
birthday. This paper is also dedicated to the memory of Vladimir Zakharov who
has recently left us.}

\section{Introduction}

We are concerned with the Cauchy problem for the Korteweg-de Vries (KdV)
equation%
\begin{equation}%
\begin{cases}
\partial_{t}q-6q\partial_{x}q+\partial_{x}^{3}q=0,\ \ \ x\in\mathbb{R}%
,t\geq0\\
q(x,0)=q(x).
\end{cases}
\label{KdV}%
\end{equation}
As is well-known, (\ref{KdV}) is the first nonlinear evolution PDE solved in
the seminal 1967 Gardner-Greene-Kruskal-Miura paper \cite{GGKM67} by the
method which is now referred to as the inverse scattering transform (IST).
Conceptually, the IST is similar to the Fourier method but is based on the
direct/inverse scattering (spectral) theory for the 1D Schr\"{o}dinger
operator $\mathbb{L}_{q}=-\partial_{x}^{2}+q(x)$. Explicit formulas, however,
are in short supply and trace formulas are among a few available.
Historically, for short-range potentials $q\left(  x\right)  $ (i.e. summable
with the first moment) such a formula (see (\ref{DT trace})) was put forward
by Deift-Trubowitz in \cite{Deift79} in the late 70s (we call it the
Deift-Trubowitz trace formula). However, no adaptation of the trace formula
(\ref{DT trace}) to the solution $q\left(  x,t\right)  $ to (\ref{KdV}) is
offered in \cite{Deift79} and, to the best of our knowledge, it has not been
done in the literature. The main goal of our contribution is to address this problem.

To this end, we first put forward an elementary approach to generate trace
formulas for the Schr\"{o}dinger operator $\mathbb{L}_{q}$ with a decaying
(but not necessarily short-range) potential $q$. More precisely, we start out
with considering potentials $q\in L^{1}\left(  \mathbb{R}\right)  $ such that
the right Jost solution $\psi(x,k)$ of $\mathbb{L}_{q}\psi=k^{2}\psi$
satisfies the condition: for all real $x$%
\[
2\mathrm{i}k\left(  e^{-\mathrm{i}kx}\psi(x,k)-1\right)  +\int_{x}^{\infty
}q\left(  s\right)  \mathrm{d}s\in H^{2},
\]
where $H^{2}$ is the usual Hardy space consisting of analytic functions on the
upper half plane with $L^{2}$ non-tangential boundary values on the real line.
We show that for any $\alpha>0$ and almost every $x$%
\begin{equation}
q\left(  x\right)  =\partial_{x}\int_{\mathbb{R}}\operatorname{Re}%
\frac{1-e^{-\mathrm{i}kx}\psi(x,k)}{k+\mathrm{i}\alpha}\frac{2k\mathrm{d}%
k}{\pi}\text{,}\label{main trace formula}%
\end{equation}
the integral being absolutely convergent. (see Theorem \ref{thm on trace} for
the complete statement.) 

While the proof, based on Hardy space arguments, is totally elementary, the
formula (\ref{main trace formula}) is surprisingly convenient. First of all,
if $q$ is short range then (\ref{main trace formula}) readily recovers the
Deift-Trubowitz trace formula (\ref{DT trace}) (see the Appendix). For this
reason we call (\ref{main trace formula}) a Deift-Trubowitz-type trace formula.

The main advantage of (\ref{main trace formula}) is that it is particularly
convenient in the KdV context. We first show that under additional assumptions
on $q$ it admits various derivations (\ref{trace formula2}),
(\ref{trace formula 3}), (\ref{trace formula 4}) that serve different
purposes. In particular, (\ref{trace formula 4}) remains valid for $q\left(
x,t\right)  $ (see section \ref{sect of kdv solutions}). The problem with the
original Deift-Trubowitz trace formula (\ref{DT trace}) is that, as we show
below (see Corollary \ref{Corol on preservation}), $q\left(  x,t\right)  $
need not remain short-range for $t>0$ and therefore the approach of
\cite{Deift79}, where (\ref{DT trace}) is derived, breaks down in a serious
way. We emphasize that we actually demonstrate as a corollary that
(\ref{DT trace}) does hold for $q\left(  x,t\right)  ,t\geq0$. This also
appears to be a new result.

It should be noticed that (\ref{trace formula 3}) is well-suited for subtle
analysis of the gain of regularity (aka dispersive smoothing) phenomenon for
the KdV equation (section \ref{sect on analytic smoothing}). We study this
phenomenon in \cite{GruRybSIMA15}, \cite{GruRyb2018}, \cite{GruRybBLMS20}
where we rely on the Dyson formula (aka the second log determinant formula)
and the theory of Hankel operators for extension of the IST to initial data
$q\left(  x\right)  $ that is essentially arbitrary at $-\infty$ (but still
short-range at $+\infty$). Comparing with the Dyson formula considerations,
our trace approach (which also crucially uses Hankel operators) is more robust
for analysis of KdV solutions (see Remark \ref{remark in last section}). To
the best of our knowledge Theorem \ref{Thm on decay} is new.

Note that for periodic potentials the trace formula was studied in great
detail the 70s by McKean-Moerbeke \cite{McKean75}, Trubowitz
\cite{Trubowitz77} and many others (see e.g. \cite{GesztesyBOOK08} for a nice
historic review) before (\ref{DT trace}). It was generalized by Craig in
\cite{Craig89} in the late 80s to arbitrary bounded continuous potentials (the
so-called Craig's trace formula). In the 90s Gesztesy et al
\cite{Gesztesyetal93} - \cite{GesztesyXi96} developed a general approach to
Craig type trace formulas based on the Krein trace formula (the "true" trace
formula) under the only condition of essential boundedness from below. The
general trace formulas studied in \cite{Gesztesyetal93} - \cite{GesztesyXi96}
yield previously known ones. In the 2000s we \cite{RybSIAM2001} introduced a
new way of generating trace-type formulas that is not based upon Krein's trace
formula but rests on the Titchmarsh-Weyl theory for second order differential
equations and asymptotics of the Titchmarsh-Weyl m-function. The approach is
quite elementary and essentially free of any conditions. Recently, in Binder
et al \cite{BinderetalDuke18} Craig's trace formulas was used in the KdV
context to address some open problems related to almost periodic initial data.

The paper is organized as follows. In section \ref{notation} we introduce our
notations Section \ref{HO} is devoted to basics of Hardy spaces and Hankel
operator our approach is based upon. In section \ref{Sect Overview} we review
the classical direct/inverse scattering theory for Schrodinger operators on
the line using the language of Hankel operators. Section
\ref{sect on trace formulas} is where our trace formulas are introduced. We do
not claim their originality but believe that the approach is new. In Section
\ref{sect of kdv solutions} we derive a representation for the solution to the
KdV equation with short-range initial data. To the best of our knowledge it is
new. In the final section \ref{sect on analytic smoothing} we demonstrate how
our trace formula for the KdV\ is well-suited for the analysis of dispersive
smoothing. The approach builds upon our recent \cite{GruRybBLMS20} and
suggests an effective way to understanding how the KdV flow trades the decay
of initial data for gain of regularity. In Appendix we demonstrate that the
Deift-Trubowitz trace formula is actually a "nonlinearization" of ours.

\section{Notations\label{notation}}

Our notations are quite standard:

\begin{itemize}
\item Unless otherwise stated, all integrals are Lebesgue and, as is commonly
done, we drop limits of integration if the integral (absolutely convergent) is
over the whole line. For convergent integrals that are not absolutely
convergent we always use the Cauchy principal value%
\[
\left(  PV\right)  \int=\lim_{a\rightarrow\infty}\int_{-a}^{a}.
\]

\item $\chi_{S}$ is the characteristic function of a (measurable) set $S$.

\item As usual, $L^{p}\left(  S\right)  ,\ 0<p\leq\infty$, is the Lebesgue
space on a (measurable) set $S$. If $S=\mathbb{R}$ then we abbreviate
$L^{p}\left(  \mathbb{R}\right)  =L^{p}$. We include $L^{p}$ in the family of
weighted $L^{p}$ spaces defined by%
\[
L_{\alpha}^{p}=\left\{  f\ |\ \int\left\vert f\left(  x\right)  \right\vert
^{p}\left\langle x\right\rangle ^{\alpha}\mathrm{d}x<\infty\right\}
,\ \ \alpha>0.
\]
where $\left\langle x\right\rangle =\sqrt{1+x^{2}}$. The class $L_{1}^{1}$ is
basic to scattering theory for 1D Schr\"{o}dinger operators (short-range potentials).

\item $\left\Vert \cdot\right\Vert _{X}$ stands for a norm in a Banach space
$X$. The most common space is $X=L^{2}$. We merely write $\left\Vert
\cdot\right\Vert $ in this case and also%
\[
\left\Vert f\right\Vert ^{2}=\left\langle f,f\right\rangle \text{ where
}\left\langle f,g\right\rangle =\int f\left(  x\right)  \overline{g}\left(
x\right)  \mathrm{d}x.
\]

\item We write $x\simeq y$ if $x=Cy$ for some universal constant $C$;
$x\lesssim_{a}y$ if $x,y\geq0$ and $x\leq C\left(  a\right)  y$ with a
positive $C$ dependent on $a$. We drop $a$ if $C$ is a universal constant.

\item We do not distinguish between classical and distributional derivatives.

\item A statement $A_{\pm}$ means two separate statements: $A_{-}$ and $A_{+}$.
\end{itemize}

\section{Hardy spaces and Hankel operators\label{HO}}

To fix our notation we review some basics of Hardy spaces and Hankel operators
following \cite{Nik2002}.

A function $f$ analytic in $\mathbb{C}^{\pm}=\left\{  z\in\mathbb{C}%
:\pm\operatorname{Im}z>0\right\}  $ is in the Hardy space $H_{\pm}^{p}$ for
some $0<p\leq\infty$ if
\[
\Vert f\Vert_{H_{\pm}^{p}}^{p}\overset{\operatorname*{def}}{=}\sup_{y>0}\Vert
f(\cdot\pm iy)\Vert_{p}<\infty.
\]
We set $H^{p}=H_{+}^{p}.$ It is a fundamental fact of the theory of Hardy
spaces that any $f\left(  z\right)  \in H_{\pm}^{p}$ with $0<p\leq\infty$ has
non-tangential boundary values $f\left(  x\pm\mathrm{i}0\right)  $ for almost
every (a.e.) $x\in\mathbb{R}$ and%
\begin{equation}
\Vert f\Vert_{H_{\pm}^{p}}=\Vert f\left(  \cdot\pm\mathrm{i}0\right)
\Vert_{L^{p}}=\left\Vert f\right\Vert _{L^{p}}. \label{H^p norm}%
\end{equation}
Classes $H_{\pm}^{\infty}$ and $H_{\pm}^{2}$ will be particularly important.
$H_{\pm}^{\infty}$ is the algebra of uniformly bounded in $\mathbb{C}^{\pm}$
functions and $H_{\pm}^{2}$ is the Hilbert space with the inner product
induced from $L^{2}$.

It is well-known that $L^{2}=H_{+}^{2}\oplus H_{-}^{2},$ the orthogonal
(Riesz) projection $\mathbb{P}_{\pm}$ onto $H_{\pm}^{2}$ being given by%
\begin{equation}
(\mathbb{P}_{\pm}f)(x)=\pm\frac{1}{2\pi\mathrm{i}}\lim_{\varepsilon
\rightarrow0+}\int\frac{f(s)\mathrm{d}s}{s-(x\pm\mathrm{i}\varepsilon)}%
=:\pm\frac{1}{2\pi\mathrm{i}}\int\frac{f(s)\ \mathrm{d}s}{s-(x\pm\mathrm{i}%
0)}. \label{proj}%
\end{equation}
Observe that the Riesz projections can also be rewritten in the form%
\begin{equation}
\left(  \mathbb{P}_{\pm}f\right)  \left(  x\right)  =\left(
\widetilde{\mathbb{P}}_{\pm}f\right)  \left(  x\right)  \mp\frac{1}%
{2\pi\mathrm{i}}\int\frac{f(s)}{s+\mathrm{i}}\mathrm{d}s, \label{regul proj}%
\end{equation}
where%
\[
\left(  \widetilde{\mathbb{P}}_{\pm}f\right)  \left(  x\right)
:=(x+\mathrm{i})\left(  \mathbb{P}_{\pm}\frac{f}{\cdot+\mathrm{i}}\right)
(x)
\]
is well-defined for any $f\in L^{\infty}$. This representation is very
important in what follows.

If $f\in L^{2}$ then $\mathbb{P}_{-}f$ is by definition in $H_{-}^{2}$ but of
course not in $L^{1}$. However under a stronger decay condition we have the
following statement.

\begin{lemma}
\label{lemma on PV}If $\left\langle x\right\rangle f\left(  x\right)  \in
L^{2}$ then%
\begin{equation}
\left(  PV\right)  \int\mathbb{P}_{-}f=\frac{1}{2}\int f. \label{PV}%
\end{equation}

\end{lemma}

\begin{proof}
Note first that if $\left\langle x\right\rangle f\left(  x\right)  \in L^{2}$
then $f$ is of course integrable as one sees from%
\[
\int\left\vert f\right\vert =\int\left\vert \left\langle x\right\rangle
f\left(  x\right)  \right\vert \frac{\mathrm{d}x}{\left\langle x\right\rangle
}\leq\left\Vert \left\langle \cdot\right\rangle f\right\Vert \left\Vert
\left\langle \cdot\right\rangle ^{-1}\right\Vert <\infty.
\]
It follows then that for a finite $a>0$ we have%
\begin{align}
\int_{-a}^{a}\mathbb{P}_{-}f  &  =\left\langle \mathbb{P}_{-}f,\chi
_{\left\vert \cdot\right\vert \leq a}\right\rangle =\left\langle
f,\mathbb{P}_{-}\chi_{\left\vert \cdot\right\vert \leq a}\right\rangle \text{
\ (by (\ref{regul proj}))}\nonumber\\
&  =\left\langle f,(\cdot+\mathrm{i})\mathbb{P}_{-}\frac{\chi_{\left\vert
\cdot\right\vert \leq a}}{\cdot+\mathrm{i}}\right\rangle +\frac{1}%
{2\pi\mathrm{i}}\int\frac{\chi_{\left\vert \cdot\right\vert \leq a}%
}{s-\mathrm{i}}\mathrm{d}s\ \int f\label{Int from -a to a}\\
&  =\left\langle (\cdot-\mathrm{i})f,\mathbb{P}_{-}\frac{\chi_{\left\vert
\cdot\right\vert \leq a}}{\cdot+\mathrm{i}}\right\rangle +\left(  \frac{1}%
{2}-\frac{1}{\pi}\arctan\frac{1}{a}\right)  \int f.\nonumber
\end{align}
Here we have used%
\begin{align}
\frac{1}{2\pi\mathrm{i}}\int\frac{\chi_{\left\vert \cdot\right\vert \leq a}%
}{s-\mathrm{i}}\mathrm{d}s  &  =\frac{1}{2\pi\mathrm{i}}\int_{-a}^{a}%
\frac{\mathrm{d}s}{s-\mathrm{i}}=\frac{1}{2}-\frac{1}{\pi}\arctan\frac{1}%
{a}\label{integral}\\
&  \rightarrow\frac{1}{2},\ \ \ a\rightarrow+\infty.\nonumber
\end{align}
Since $\chi_{\left\vert \cdot\right\vert \leq a}=1-\chi_{\left\vert
\cdot\right\vert >a}$ and $1/\left(  x+\mathrm{i}\right)  \in H^{2}$ we have%
\[
\mathbb{P}_{-}\frac{\chi_{\left\vert \cdot\right\vert \leq a}}{\cdot
+\mathrm{i}}=\mathbb{P}_{-}\frac{1}{\cdot+\mathrm{i}}-\mathbb{P}_{-}\frac
{\chi_{\left\vert \cdot\right\vert >a}}{\cdot+\mathrm{i}}=-\mathbb{P}_{-}%
\frac{\chi_{\left\vert \cdot\right\vert >a}}{\cdot+\mathrm{i}}%
\]
and therefore%
\[
\left\langle (\cdot-\mathrm{i})f,\mathbb{P}_{-}\frac{\chi_{\left\vert
\cdot\right\vert \leq a}}{\cdot+\mathrm{i}}\right\rangle =-\left\langle
(\cdot-\mathrm{i})f,\mathbb{P}_{-}\frac{\chi_{\left\vert \cdot\right\vert >a}%
}{\cdot+\mathrm{i}}\right\rangle .
\]
It follows that%
\begin{align*}
\left\vert \left\langle (\cdot-\mathrm{i})f,\mathbb{P}_{-}\frac{\chi
_{\left\vert \cdot\right\vert \leq a}}{\cdot+\mathrm{i}}\right\rangle
\right\vert  &  \leq\left\Vert (\cdot-\mathrm{i})f\right\Vert \left\Vert
\mathbb{P}_{-}\frac{\chi_{\left\vert \cdot\right\vert >a}}{\cdot+\mathrm{i}%
}\right\Vert \\
&  \leq\left\Vert \left\langle \cdot\right\rangle f\right\Vert \left\Vert
\frac{\chi_{\left\vert \cdot\right\vert >a}}{\cdot+\mathrm{i}}\right\Vert
\rightarrow0,\ a\rightarrow\infty.
\end{align*}
This means that%
\[
\lim_{a\rightarrow\infty}\left\langle (\cdot-\mathrm{i})f,\mathbb{P}_{-}%
\frac{\chi_{\left\vert \cdot\right\vert \leq a}}{\cdot+\mathrm{i}%
}\right\rangle =0
\]
and we can pass in (\ref{Int from -a to a}) as $a\rightarrow\infty$:%
\[
\lim_{a\rightarrow\infty}\int_{-a}^{a}\mathbb{P}_{-}f=\frac{1}{2}\int f.
\]

\end{proof}

We now define the Hankel operator on $H^{2}$. Let $(\mathbb{J}f)(x)=f(-x)$ be
the operator of reflection. Given $\varphi\in L^{\infty}$ the operator
$\mathbb{H}(\varphi):H^{2}\rightarrow H^{2}$ given by the formula%
\begin{equation}
\mathbb{H}(\varphi)f=\mathbb{JP}_{-}\varphi f,\ \ \ f\in H_{+}^{2},
\label{Hankel}%
\end{equation}
is called the Hankel\emph{ }operator with symbol $\varphi$. Clearly
$\left\Vert \mathbb{H}(\varphi)\right\Vert \leq\left\Vert \varphi\right\Vert
_{L^{\infty}}$, $\mathbb{H}(\varphi)$ is self-adjoint if $(\mathbb{J}%
\varphi)(x)=\overline{\varphi\left(  x\right)  }$ (this is always our case),
$\mathbb{H}(\varphi)=0$ if $\varphi$ is a constant, and%
\begin{align}
\mathbb{H}(\varphi)  &  =\mathbb{H}(\widetilde{\mathbb{P}}_{-}\varphi
)\label{regularized Hankel}\\
&  =\mathbb{H}(\mathbb{P}_{-}\varphi)\text{ (if }\varphi\in L^{2}\cap
L^{\infty}\text{).}\nonumber
\end{align}
The relevance of the Hankel operator in our setting is on the surface as the
Marchenko operator, the cornerstone of the IST, is a Hankel operator. However,
while in the literature on integrable systems it is rarely used in the form
(\ref{Hankel}), we find it particularly convenient due, among others, to the
property (\ref{regularized Hankel}), which is less transparent in the integral representation.

Finally we note that reliance on the theory of Hankel operator in the study of
completely integrable systems has recently picked up momentum (see e.g.
\cite{Blower23}, \cite{Doikouetal21}, \cite{GerardPushnitski2023},
\cite{Gruetal23}, \cite{Malham22} and the references cited therein).

\section{Overview of short-range scattering\label{Sect Overview}}

Unless otherwise stated all facts are taken from \cite{MarchBook2011}. Through
this section we assume that $q$ is short-range, i.e. $q\in L_{1}^{1}$.
Associate with $q$ the full line Schr\"{o}dinger operator $\mathbb{L}%
_{q}=-\partial_{x}^{2}+q(x)$. As is well-known, $\mathbb{L}_{q}$ is
self-adjoint on $L^{2}$ and its spectrum consists of $J$ simple negative
eigenvalues $\{-\kappa_{j}^{2}:1\leq j\leq J\},$ called bound states ($J=0$ if
there are no bound states), and two fold absolutely continuous component
filling $\left(  0,\infty\right)  $. There is no singular continuous spectrum.
Two linearly independent (generalized) eigenfunctions of the a.c. spectrum
$\psi_{\pm}(x,k),\;k\in\mathbb{R}$, can be chosen to satisfy
\begin{equation}
\psi_{\pm}(x,k)=e^{\pm\mathrm{i}kx}+o(1),\;\partial_{x}\psi_{\pm}%
(x,k)\mp\mathrm{i}k\psi_{\pm}(x,k)=o(1),\ \ x\rightarrow\pm\infty.
\label{eq6.2}%
\end{equation}
The function $\psi_{\pm}$, referred to as right/left Jost solution of the
Schr\"{o}dinger equation%
\begin{equation}
\mathbb{L}_{q}\psi=k^{2}\psi, \label{eq6.3}%
\end{equation}
is analytic for $\operatorname{Im}k>0$. It is convenient to introduce%
\[
y_{\pm}\left(  k,x\right)  :=e^{\mp\mathrm{i}kx}\psi_{\pm}\left(  x,k\right)
-1,
\]
($1+y_{\pm}\left(  k,x\right)  $ is sometimes referred to as the Faddeev
function), which is $H^{2}$ for each $x$. Since $q$ is real, $\overline{\psi
}_{\pm}$ also solves (\ref{eq6.3}) and one can easily see that the pairs
$\{\psi_{+},\overline{\psi}_{+}\}$ and $\{\psi_{-},\overline{\psi}_{-}\}$ form
fundamental sets for (\ref{eq6.3}). Hence $\psi_{\mp}$ is a linear combination
of $\{\psi_{\pm},\overline{\psi}_{\pm}\}$. We write this fact as follows%
\begin{equation}
T(k)\psi_{\mp}(x,k)=\overline{\psi_{\pm}(x,k)}+R_{\pm}(k)\psi_{\pm
}(x,k),\ \ \ k\in\mathbb{R}\text{,} \label{R basic scatt identity}%
\end{equation}
where $T$ and $R_{\pm},$ are called transmission, right/left reflection
coefficients respectively. The function $T\left(  k\right)  $ is meromorphic
for $\operatorname{Im}k>0$ with simple poles at $\left(  \mathrm{i}\kappa
_{j}\right)  $ and continuous for $\operatorname{Im}k=0$. Generically,
$T\left(  0\right)  =0$. The reflection coefficient $R_{\pm}\left(  k\right)
\in L^{2}$ but need not admit be analytic.

In the context of the IST Zakharov-Faddeev trace formulas \cite{Zakharov71}
(conservation laws) play very important role. For Schwarz potentials $q$ they
are infinitely many. Explicitly,%
\begin{equation}
\frac{8}{\pi}\int\log\left(  1-\left\vert R\left(  k\right)  \right\vert
^{2}\right)  ^{-1}\mathrm{d}k=\int q+%
%TCIMACRO{\dsum }%
%BeginExpansion
{\displaystyle\sum}
%EndExpansion
\kappa_{n}\text{ \ \ (first trace formula)} \label{ZF trace 1}%
\end{equation}%
\begin{equation}
\frac{8}{\pi}\int k^{2}\log\left(  1-\left\vert R_{\pm}\left(  k\right)
\right\vert ^{2}\right)  ^{-1}\mathrm{d}k=\int q^{2}-\frac{16}{3}%
%TCIMACRO{\dsum }%
%BeginExpansion
{\displaystyle\sum}
%EndExpansion
\kappa_{n}^{3}\text{ \ \ (second trace formula)} \label{ZF trace 2}%
\end{equation}
It is shown in the recent \cite{Hryniv21} that (\ref{ZF trace 1}) holds for
any $q\in L^{1}$, each term being finite. Since $\left\vert R_{\pm}\left(
k\right)  \right\vert \leq1$ and%
\[
\log\left(  1-\left\vert R_{\pm}\left(  k\right)  \right\vert ^{2}\right)
^{-1}\geq\left\vert R_{\pm}\left(  k\right)  \right\vert ^{2},
\]
one concludes that $R_{\pm}\left(  k\right)  \in L^{2}$ for $q\in L^{1}$. The
second one (\ref{ZF trace 2}) holds for $q\in L^{1}\cap L^{2}$
\cite{KillipSimon2008} and readily implies that $kR_{\pm}\left(  k\right)  \in
L^{2}$. Note, that Zakharov-Faddeev trace formulas are not directly related to
the trace formulas we discuss in Introduction but they are also related to the
trace of some operators.

The identities (\ref{R basic scatt identity}) are totally elementary but serve
as a basis for inverse scattering theory and for this reason they are commonly
referred to as basic scattering relations. As is well-known (see, e.g.
\cite{MarchBook2011}), the triple $\{R_{\pm},(\kappa_{j},c_{\pm,j})\}$, where
$c_{\pm,j}=\left\Vert \psi_{\pm}(\cdot,\mathrm{i}\kappa_{j})\right\Vert ^{-1}%
$, determines $q$ uniquely and is called the scattering data for
$\mathbb{L}_{q}$. We emphasize that in order to come from a $L_{1}^{1}$
potential the scattering data $\{R_{\pm},(\kappa_{n},c_{\pm,n})\}$ must
satisfy some conditions known as Marchenko's characterization
\cite{MarchBook2011}. The actual process of solving the inverse scattering
problem necessary for the IST is historically based on the Marchenko theory
(also knows as Faddeev-Marchenko or Gelfand-Levitan-Marchenko). In fact, this
procedure is quite transparent from the Hankel operator point of view. Indeed,
replacing $\psi_{\pm}$ in (\ref{R basic scatt identity}) with $y_{\pm}$ and
applying the operator $\mathbb{JP}_{-}$, a straightforward computation
\cite{GruRybSIMA15} leads to%
\begin{equation}
y_{\pm}+\mathbb{H}(\varphi_{\pm})y_{\pm}=-\mathbb{H}(\varphi_{\pm})1,\text{
(Marchenko's equation)} \label{Marchenko eq}%
\end{equation}
where $\mathbb{H}(\varphi_{\pm})$ is the Hankel operator (\ref{Hankel}) with
symbol%
\begin{equation}
\varphi_{\pm}\left(  k,x\right)  =%
%TCIMACRO{\dsum _{n=1}^{N}}%
%BeginExpansion
{\displaystyle\sum_{n=1}^{N}}
%EndExpansion
\frac{-\mathrm{i}c_{\pm,n}^{2}e^{\mp2\kappa_{n}x}}{k-\mathrm{i}\kappa_{n}%
}+R_{\pm}\left(  k\right)  e^{\pm2\mathrm{i}kx}, \label{Symbol}%
\end{equation}
and $\mathbb{H}(\varphi_{\pm})1$ is understood as%
\[
\mathbb{H}(\varphi_{\pm})1=\mathbb{JP}_{-}\varphi_{\pm}=\mathbb{P}%
_{+}\mathbb{J}\varphi_{\pm}=\mathbb{P}_{+}\overline{\varphi}_{\pm}.
\]
We call (\ref{Marchenko eq}) the Marchenko equation as its Fourier image is
the Marchenko integral equation. It is proven in \cite[Theorem 8.2]%
{GruRybSIMA15} that $I+\mathbb{H}(\varphi_{\pm})$ is positive definite and
therefore%
\begin{equation}
y_{\pm}=-\left[  I+\mathbb{H}(\varphi_{\pm})\right]  ^{-1}\mathbb{H}%
(\varphi_{\pm})1\in H^{2}. \label{sltn to March eq}%
\end{equation}
Thus, given data $\{R_{\pm},(\kappa_{j},c_{\pm,j})\}$ we compute $\varphi
_{\pm}$ by (\ref{Symbol}) and form the Hankel operator\ $\mathbb{H}%
(\varphi_{\pm})$. The function $y_{\pm}\left(  k,x\right)  $ is found by
(\ref{sltn to March eq}). The potential $q\left(  x\right)  $ can then be
recovered in a few ways. Our method is, of course, to apply a suitable trace
formula, which we derive in the next section.

Since many of our proofs below are based on limiting arguments we need to
understand in what sense scattering data converges as we approximate $q$ in
the $L_{1}^{1}$. In particular the following statement plays an important role.

\begin{proposition}
\label{Prop on R}If $q_{n}\left(  x\right)  $ converges in $L_{1}^{1}$ to
$q\left(  x\right)  $ then the sequence of reflection coefficients $R_{\pm
,n}\left(  k\right)  $ corresponding to $q_{n}\left(  x\right)  $ converges in
$L^{2}$ to $R_{\pm}\left(  k\right)  .$
\end{proposition}

\begin{proof}
We consider the $+$ case only and we suppress $+$ sign. We use the following a
priori estimates (see e.g. \cite{Deift79})%
\begin{equation}
\left\vert y_{-}\left(  x,k\right)  \right\vert \lesssim_{q}\left\langle
x\right\rangle /\left\langle k\right\rangle \label{est for y}%
\end{equation}%
\begin{equation}
\left\vert y_{-}\left(  x,k\right)  -y_{-,n}\left(  x,k\right)  \right\vert
\lesssim_{q}\left\langle x\right\rangle \left\Vert q-q_{n}\right\Vert
_{L_{1}^{1}} \label{est for y-yn}%
\end{equation}%
\begin{equation}
\left\vert T\left(  k\right)  -T_{n}\left(  k\right)  \right\vert \lesssim
_{q}\left\vert k\right\vert ^{-1}\left\Vert q-q_{n}\right\Vert _{L_{1}^{1}}.
\label{est for T-Tn}%
\end{equation}

Consider $\left\Vert R-R_{n}\right\Vert ^{2}$ and rewrite it as ($\varepsilon$
is any)%
\begin{align}
\left\Vert R-R_{n}\right\Vert ^{2}  &  =\left\Vert \left(  R-R_{n}\right)
\chi_{\left\vert \cdot\right\vert \leq\varepsilon}\right\Vert ^{2}+\left\Vert
\left(  R-R_{n}\right)  \chi_{\left\vert \cdot\right\vert >\varepsilon
}\right\Vert ^{2}\label{est for L2 norm of R-Rn}\\
&  \leq8\varepsilon+\left\Vert \left(  R-R_{n}\right)  \chi_{\left\vert
\cdot\right\vert >\varepsilon}\right\Vert ^{2}.\nonumber
\end{align}
It follows from the general formula \cite{Deift79}%
\begin{equation}
R\left(  k\right)  =\frac{T\left(  k\right)  }{2\mathrm{i}k}\int
e^{-2\mathrm{i}kx}q\left(  x\right)  \left(  1+y_{-}\left(  x,k\right)
\right)  \mathrm{d}x \label{rep for R}%
\end{equation}
that%
\begin{align*}
R\left(  k\right)  -R_{n}\left(  k\right)   &  =\frac{T\left(  k\right)
-T_{n}\left(  k\right)  }{2\mathrm{i}k}\int e^{-2\mathrm{i}kx}q\left(
x\right)  \left(  1+y_{-}\left(  x,k\right)  \right)  \mathrm{d}x\\
&  +\frac{T_{n}\left(  k\right)  }{2\mathrm{i}k}\int e^{-2\mathrm{i}kx}\left(
q\left(  x\right)  -q_{n}\left(  x\right)  \right)  \mathrm{d}x\\
&  +\frac{T_{n}\left(  k\right)  }{2\mathrm{i}k}\int e^{-2\mathrm{i}%
kx}q\left(  x\right)  \left(  y_{-}\left(  x,k\right)  -y_{-,n}\left(
x,k\right)  \right)  \mathrm{d}x\\
&  =I_{1}\left(  k\right)  +I_{2}\left(  k\right)  +I_{3}\left(  k\right)
\end{align*}
and hence%
\[
\left\Vert \left(  R-R_{n}\right)  \chi_{\left\vert \cdot\right\vert
>\varepsilon}\right\Vert \leq\left\Vert I_{1}\chi_{\left\vert \cdot\right\vert
>\varepsilon}\right\Vert +\left\Vert I_{2}\chi_{\left\vert \cdot\right\vert
>\varepsilon}\right\Vert +\left\Vert I_{3}\chi_{\left\vert \cdot\right\vert
>\varepsilon}\right\Vert .
\]
Estimate each term separately. For $\left\Vert I_{1}\chi_{\left\vert
\cdot\right\vert >\varepsilon}\right\Vert $ we have%
\begin{align*}
\left\Vert I_{1}\chi_{\left\vert \cdot\right\vert >\varepsilon}\right\Vert
^{2}  &  \lesssim_{q}\int\left\vert q\right\vert \left(  1+\left\vert
y_{-}\left(  x,k\right)  \right\vert \right)  \mathrm{d}x\cdot\int_{\left\vert
k\right\vert >\varepsilon}\left\vert \frac{T\left(  k\right)  -T_{n}\left(
k\right)  }{k}\right\vert ^{2}\mathrm{d}k\\
&  \lesssim_{q}\left\Vert q-q_{n}\right\Vert _{L_{1}^{1}}^{2}\int_{\left\vert
k\right\vert >\varepsilon}k^{-4}\mathrm{d}k\text{\ \ (by (\ref{est for y}%
),(\ref{est for T-Tn}))}\\
&  \lesssim\varepsilon^{-3}\left\Vert q-q_{n}\right\Vert _{L_{1}^{1}}^{2}.
\end{align*}
Thus%
\[
\left\Vert I_{1}\chi_{\left\vert \cdot\right\vert >\varepsilon}\right\Vert
\lesssim_{q}\varepsilon^{-3/2}\left\Vert q-q_{n}\right\Vert _{L_{1}^{1}}.
\]
For $\left\Vert I_{2}\chi_{\left\vert \cdot\right\vert >\varepsilon
}\right\Vert $ we have%
\begin{align*}
\left\Vert I_{2}\chi_{\left\vert \cdot\right\vert >\varepsilon}\right\Vert
^{2}  &  \leq\left\Vert q-q_{n}\right\Vert _{L^{1}}^{2}\int_{\left\vert
k\right\vert >\varepsilon}\left\vert \frac{T_{n}\left(  k\right)  }%
{k}\right\vert ^{2}\mathrm{d}k\\
&  \lesssim\frac{1}{\varepsilon}\left\Vert q-q_{n}\right\Vert _{L_{1}^{1}}^{2}%
\end{align*}
and hence%
\[
\left\Vert I_{2}\chi_{\left\vert \cdot\right\vert >\varepsilon}\right\Vert
\lesssim_{q}\varepsilon^{-1/2}\left\Vert q-q_{n}\right\Vert _{L^{1}}.
\]
Finally for $\left\Vert I_{2}\chi_{\left\vert \cdot\right\vert >\varepsilon
}\right\Vert $ one has in a similar manner%
\begin{align*}
\left\Vert I_{3}\chi_{\left\vert \cdot\right\vert >\varepsilon}\right\Vert  &
\leq_{q}\varepsilon^{-1/2}\sup_{k}\left\vert \int e^{-2\mathrm{i}kx}q\left(
x\right)  \left(  y_{-}\left(  x,k\right)  -y_{-,n}\left(  x,k\right)
\right)  \mathrm{d}x\right\vert \\
&  \leq_{q}\varepsilon^{-1/2}\left\Vert q-q_{n}\right\Vert _{L_{1}^{1}%
}\ \text{\ (by (\ref{est for y-yn}))}%
\end{align*}
and hence%
\[
\left\Vert I_{3}\chi_{\left\vert \cdot\right\vert >\varepsilon}\right\Vert
\lesssim_{q}\varepsilon^{-1/2}\left\Vert q-q_{n}\right\Vert _{L_{1}^{1}}.
\]
One can now sees that each $\left\Vert I_{j}\chi_{\left\vert \cdot\right\vert
>\varepsilon}\right\Vert $, $j=1,2,3,$ vanishes as $\left\Vert q-q_{n}%
\right\Vert _{L_{1}^{1}}$ does and hence, since $\varepsilon$ is arbitrary, it
follows from (\ref{est for L2 norm of R-Rn}) that%
\[
\left\Vert R-R_{n}\right\Vert \rightarrow0,n\rightarrow\infty.
\]

\end{proof}

Note that the question in what sense the reflection coefficient converges when
we approximate the potential in a certain way is a subtle one
\cite{RemlingCMP15}.

Finally we observe that $\psi_{\pm},y_{\pm},T,R_{\pm}$ as functions of $k$
(momentum) satisfy%
\begin{equation}
\left(  \mathbb{J}f\right)  \left(  k\right)  =f\left(  -k\right)
=\overline{f}\left(  k\right)  \text{ \ \ \ (symmetry property).} \label{symm}%
\end{equation}

\section{Trace formulas\label{sect on trace formulas}}

In this section we put forward a new approach to generate Deift-Trubowitz type
trace formulas. It is based on Hardy spaces and Hankel operators.

\begin{theorem}
\label{thm on trace} Suppose that $q\in L^{1}$ and%
\[
Q_{+}\left(  x\right)  :=\int_{x}^{\infty}q\left(  s\right)  \mathrm{d}%
s,Q_{-}\left(  x\right)  :=\int_{-\infty}^{x}q\left(  s\right)  \mathrm{d}s.
\]
Let $\psi_{\pm}(x,k)$ be right/left Jost solution and%
\[
y_{\pm}\left(  k,x\right)  =e^{\mp\mathrm{i}kx}\psi_{\pm}(x,k)-1.
\]
If for all real $x$%
\begin{equation}
2\mathrm{i}ky_{\pm}\left(  k,x\right)  +Q_{\pm}\left(  x\right)  \in H^{2},
\label{cond on y}%
\end{equation}
then for any $\alpha>0$ and a.e. $x$%
\begin{equation}
q\left(  x\right)  =\mp\frac{2}{\pi}\partial_{x}\int\operatorname{Re}%
\frac{y_{\pm}\left(  k,x\right)  }{k+\mathrm{i}\alpha}k\mathrm{d}k\text{
(trace formula).} \label{trace formula 1}%
\end{equation}
If for every real $x$%
\begin{equation}
y_{\pm}\left(  \cdot,x\right)  \in H^{2} \label{cond on y 1}%
\end{equation}
then (\ref{trace formula 1}) simplifies to read%
\begin{equation}
q\left(  x\right)  =\mp\frac{2}{\pi}\partial_{x}\int\operatorname{Re}y_{\pm
}\left(  k,x\right)  \mathrm{d}k. \label{trace formula2}%
\end{equation}

\end{theorem}

\begin{proof}
Note first that both Jost solutions exist for $q\in L^{1}$ (not only for
$L_{1}^{1}$). Multiplying (\ref{cond on y}) by $\mathrm{i}/\left(
k+\mathrm{i}\alpha\right)  \in H^{2}$ ($\alpha>0$) and recalling that a
product of two $H^{2}$ functions is in $H^{1}$, we have%
\begin{equation}
\frac{2k}{k+\mathrm{i}\alpha}y_{\pm}\left(  k,x\right)  -\frac{\mathrm{i}%
}{k+\mathrm{i}\alpha}Q_{\pm}\left(  x\right)  \in H^{1}. \label{eq 1}%
\end{equation}
But it is well-known that%
\begin{equation}
f\in H^{1}\Longrightarrow\int f\left(  k+\mathrm{i}0\right)  \mathrm{d}k=0
\label{eq 2}%
\end{equation}
and therefore%
\[
\int\left[  \frac{2k}{k+\mathrm{i}\alpha}y_{\pm}\left(  k,x\right)
-\frac{\mathrm{i}}{k+\mathrm{i}\alpha}Q_{\pm}\left(  x\right)  \right]
\mathrm{d}k=0.
\]
For its real part we have%
\[
\int\left\{  \operatorname{Re}\left[  \frac{2k}{k+\mathrm{i}\alpha}y_{\pm
}\left(  k,x\right)  \right]  -\frac{\alpha}{k^{2}+\alpha^{2}}Q_{\pm}\left(
x\right)  \right\}  \mathrm{d}k=0,
\]
which can be rearranged to read%
\[
\pi Q_{\pm}\left(  x\right)  =2\int\operatorname{Re}\left[  \frac{y_{\pm
}\left(  k,x\right)  }{k+\mathrm{i}\alpha}\right]  k\mathrm{d}k
\]
and (\ref{trace formula 1}) follows upon differentiating in $x$.

We show now (\ref{trace formula2}). To this end, we just split
(\ref{trace formula 1}) as%
\begin{align*}
\int\operatorname{Re}\frac{ky_{\pm}\left(  k,x\right)  }{k+\mathrm{i}\alpha
}\mathrm{d}k  &  =\int\operatorname{Re}y_{\pm}\left(  k,x\right)
\mathrm{d}k\\
&  +\alpha\operatorname{Im}\int\frac{y_{\pm}\left(  k,x\right)  }%
{k+\mathrm{i}\alpha}\mathrm{d}k
\end{align*}
and observe that by (\ref{eq 2}) the second integral on the right hand side is
zero (both $y_{\pm}$ and $1/\left(  k+\mathrm{i}\alpha\right)  $ are in
$H^{2}$).
\end{proof}

\begin{remark}
Under the condition $q\in L_{1}^{1}$ the following formula is proven in
\cite{Deift79} (only $+$ sign is considered):%
\begin{align}
q\left(  x\right)   &  =-4%
%TCIMACRO{\dsum _{n=1}^{N}}%
%BeginExpansion
{\displaystyle\sum_{n=1}^{N}}
%EndExpansion
\kappa_{n}c_{+,n}^{2}\psi_{+}\left(  x,\mathrm{i}\kappa_{n}\right)
^{2}\label{DT trace}\\
&  +\frac{2\mathrm{i}}{\pi}\left(  PV\right)  \int R_{+}\left(  k\right)
\psi_{+}\left(  x,k\right)  ^{2}k\mathrm{d}k.\text{ (Deift-Trubowitz trace
formula)}\nonumber
\end{align}
Visually it is very different from (\ref{trace formula2}) ($\psi_{+}\left(
x,k\right)  $ appears in (\ref{DT trace}) squared whereas in
(\ref{trace formula2}) it does not). One can however show that
(\ref{trace formula2}) implies (\ref{DT trace}). We demonstrate this fact in
the Appendix. Theorem \ref{thm on trace} is an extension of (\ref{DT trace})
as it accepts certain singularities of $\psi_{\pm}\left(  k,x\right)  $ at
$k=0$. The latter may occur if $q\notin L_{1}^{1}$. Thus following the
terminology of \cite{Deift79} we may refer to our (\ref{trace formula 1}) and
(\ref{trace formula2}) as trace formulas.
\end{remark}

The next statement offers a version of (\ref{DT trace}) that is linear with
respect to the Jost solution $\psi_{\pm}$.

\begin{corollary}
\label{Corollary 1}Suppose $q\in L_{1}^{1}\cap L^{2}$ and let $\left\{
R_{\pm},\kappa_{n},c_{\pm,n}\right\}  $ be its scattering data. Then%
\begin{equation}
q\left(  x\right)  =\pm\partial_{x}\left\{  2%
%TCIMACRO{\dsum _{n=1}^{N}}%
%BeginExpansion
{\displaystyle\sum_{n=1}^{N}}
%EndExpansion
c_{\pm,n}^{2}e^{\mp\kappa_{n}x}\psi_{\pm}(x,\mathrm{i}\kappa_{n})+\frac{1}%
{\pi}\int e^{\pm\mathrm{i}kx}R_{\pm}\left(  k\right)  \psi_{\pm}%
(x,k)\mathrm{d}k\right\}  . \label{trace formula 3}%
\end{equation}

\end{corollary}

\begin{proof}
As is well-known (see e.g. \cite{MarchBook2011}), for $q\in L^{1}$%
\begin{equation}
y_{\pm}\left(  k,x\right)  =\frac{\mathrm{i}}{2k}Q_{\pm}\left(  x\right)
+O\left(  k^{-2}\right)  ,\ \ \ k\rightarrow\pm\infty, \label{asympt}%
\end{equation}
and furthermore $y_{\pm}\left(  k,x\right)  $ is bounded at $k=0$ for $q\in
L_{1}^{1}$. It immediately follows that the condition (\ref{cond on y}) is
satisfied. Also, the condition (\ref{cond on y 1}) holds due to
(\ref{sltn to March eq}). Therefore (\ref{trace formula2}) holds for
short-range $q$. To show (\ref{trace formula 3}) we turn to the Marchenko
equation (\ref{Marchenko eq}). Applying the operator of reflection
$\mathbb{J}$ to this equation and recalling the symmetry property (\ref{symm})
we have%
\[
\overline{y}_{\pm}+\mathbb{P}_{-}(\varphi_{\pm}y_{\pm})=-\mathbb{P}_{-}%
\varphi_{\pm},
\]
which together with (\ref{Marchenko eq}) yield%
\[
2\operatorname{Re}y_{\pm}=-\mathbb{JP}_{-}(\varphi_{\pm}y_{\pm})-\mathbb{P}%
_{-}(\varphi_{\pm}y_{\pm})-\mathbb{JP}_{-}\varphi_{\pm}-\mathbb{P}_{-}%
\varphi_{\pm}.
\]
Since obviously%
\[
\int_{-a}^{a}\mathbb{J}f=\int_{-a}^{a}f
\]
we have%
\begin{equation}
\int_{-a}^{a}\operatorname{Re}y_{\pm}=-\int_{-a}^{a}\mathbb{P}_{-}\varphi
_{\pm}-\int_{-a}^{a}\mathbb{P}_{-}\left(  \varphi_{\pm}y_{\pm}\right)  .
\label{int from -a to a}%
\end{equation}
Consider each term on the right hand side of (\ref{int from -a to a}).
Observing that $\left(  k-\mathrm{i}\kappa_{n}\right)  ^{-1}\in H_{-}^{2}$ and
hence $\mathbb{P}_{-}\left(  k-\mathrm{i}\kappa_{n}\right)  ^{-1}=\left(
k-\mathrm{i}\kappa_{n}\right)  ^{-1}$, we have%
\begin{align*}
\mathbb{P}_{-}\varphi_{\pm}  &  =\mathbb{P}_{-}\left[
%TCIMACRO{\dsum _{n=1}^{N}}%
%BeginExpansion
{\displaystyle\sum_{n=1}^{N}}
%EndExpansion
\frac{-\mathrm{i}c_{\pm,n}^{2}e^{\mp2\kappa_{n}x}}{k-\mathrm{i}\kappa_{n}%
}\right]  +\mathbb{P}_{-}\left[  R_{\pm}\left(  k\right)  e^{\pm2\mathrm{i}%
kx}\right] \\
&  =%
%TCIMACRO{\dsum _{n=1}^{N}}%
%BeginExpansion
{\displaystyle\sum_{n=1}^{N}}
%EndExpansion
\frac{-\mathrm{i}c_{\pm,n}^{2}e^{\mp2\kappa_{n}x}}{k-\mathrm{i}\kappa_{n}%
}+\mathbb{P}_{-}\left[  R_{\pm}\left(  k\right)  e^{\pm2\mathrm{i}kx}\right]
\end{align*}
and thus%
\begin{align}
\int_{-a}^{a}\mathbb{P}_{-}\varphi_{\pm}  &  =\int_{-a}^{a}%
%TCIMACRO{\dsum _{n=1}^{N}}%
%BeginExpansion
{\displaystyle\sum_{n=1}^{N}}
%EndExpansion
\frac{-\mathrm{i}c_{\pm,n}^{2}e^{\mp2\kappa_{n}x}}{k-\mathrm{i}\kappa_{n}%
}\mathrm{d}k+\int_{-a}^{a}\mathbb{P}_{-}\left[  R_{\pm}\left(  k\right)
e^{\pm2\mathrm{i}kx}\right]  \mathrm{d}k\nonumber\\
&  =-\mathrm{i}%
%TCIMACRO{\dsum _{n=1}^{N}}%
%BeginExpansion
{\displaystyle\sum_{n=1}^{N}}
%EndExpansion
c_{\pm,n}^{2}e^{\mp2\kappa_{n}x}\int_{-a}^{a}\frac{\mathrm{d}k}{k-\mathrm{i}%
\kappa_{n}}+\int_{-a}^{a}\mathbb{P}_{-}\left[  R_{\pm}\left(  k\right)
e^{\pm2\mathrm{i}kx}\right]  \mathrm{d}k. \label{prelimit}%
\end{align}
Pass in (\ref{prelimit}) now to the limit as $a\rightarrow\infty$. By
(\ref{integral})%
\[
\lim_{a\rightarrow\infty}\int_{-a}^{a}\frac{\mathrm{d}k}{k-\mathrm{i}%
\kappa_{n}}=\lim_{a\rightarrow\infty}\int_{-a}^{a}\frac{\mathrm{d}%
k}{k-\mathrm{i}}=\pi\mathrm{i}%
\]
and hence substituting this into (\ref{prelimit}) one has%
\begin{align*}
\left(  PV\right)  \int_{-a}^{a}\mathbb{P}_{-}\varphi_{\pm}  &  =\pi%
%TCIMACRO{\dsum _{n=1}^{N}}%
%BeginExpansion
{\displaystyle\sum_{n=1}^{N}}
%EndExpansion
c_{\pm,n}^{2}e^{\mp2\kappa_{n}x}\\
&  +\left(  PV\right)  \int_{-a}^{a}\mathbb{P}_{-}\left[  R_{\pm}\left(
k\right)  e^{\pm2\mathrm{i}kx}\right]  \mathrm{d}k.
\end{align*}
It remains to evaluate the integral on the right hand side. As we have shown
in Section \ref{Sect Overview}, $\left\langle k\right\rangle R_{\pm}\left(
k\right)  \in L^{2}$. By Lemma \ref{lemma on PV} then%
\[
\left(  PV\right)  \int_{-a}^{a}\mathbb{P}_{-}\left[  R_{\pm}\left(  k\right)
e^{\pm2\mathrm{i}kx}\right]  \mathrm{d}k=\frac{1}{2}\int R_{\pm}\left(
k\right)  e^{\pm2\mathrm{i}kx}\mathrm{d}k
\]
and finally%
\begin{align*}
\left(  PV\right)  \int\mathbb{P}_{-}\varphi_{\pm}  &  =\pi%
%TCIMACRO{\dsum _{n=1}^{N}}%
%BeginExpansion
{\displaystyle\sum_{n=1}^{N}}
%EndExpansion
c_{\pm,n}^{2}e^{\mp2\kappa_{n}x}\\
&  +\frac{1}{2}\int R_{\pm}\left(  k\right)  e^{\pm2\mathrm{i}kx}\mathrm{d}k.
\end{align*}
Similarly,%
\[
\left(  PV\right)  \int\mathbb{P}_{-}\left(  \varphi_{\pm}y_{\pm}\right)
=\frac{1}{2}\int\varphi_{\pm}y_{\pm}%
\]
and from (\ref{int from -a to a}) we obtain%
\begin{align*}
\int\operatorname{Re}y_{\pm}  &  =-\int\mathbb{P}_{-}\varphi_{\pm}%
-\int\mathbb{P}_{-}\left(  \varphi_{\pm}y_{\pm}\right) \\
&  =-\pi%
%TCIMACRO{\dsum _{n=1}^{N}}%
%BeginExpansion
{\displaystyle\sum_{n=1}^{N}}
%EndExpansion
c_{\pm,n}^{2}e^{\mp2\kappa_{n}x}-\frac{1}{2}\int R_{\pm}\left(  k\right)
e^{\pm2\mathrm{i}kx}\mathrm{d}k\\
&  -\frac{1}{2}\int\varphi_{\pm}y_{\pm}.
\end{align*}
Inserting this into (\ref{trace formula2}) one has%
\begin{align*}
q\left(  x\right)   &  =\mp\frac{2}{\pi}\partial_{x}\int\operatorname{Re}%
y_{\pm}\left(  k,x\right)  \mathrm{d}k\\
&  =\pm\partial_{x}\left\{  2%
%TCIMACRO{\dsum _{n=1}^{N}}%
%BeginExpansion
{\displaystyle\sum_{n=1}^{N}}
%EndExpansion
c_{\pm,n}^{2}e^{\mp2\kappa_{n}x}+\int R_{\pm}\left(  k\right)  e^{\pm
2\mathrm{i}kx}\frac{\mathrm{d}k}{\pi}+\int\varphi_{\pm}\left(  k,x\right)
y_{\pm}\left(  k,x\right)  \frac{\mathrm{d}k}{\pi}\right\}  .
\end{align*}
It remains to evaluate the last integral on the right hand side. By the Cauchy
formula%
\begin{align*}
\int\varphi_{\pm}y_{\pm}  &  =\int%
%TCIMACRO{\dsum _{n=1}^{N}}%
%BeginExpansion
{\displaystyle\sum_{n=1}^{N}}
%EndExpansion
\frac{-\mathrm{i}c_{\pm,n}^{2}e^{\mp2\kappa_{n}x}}{k-\mathrm{i}\kappa_{n}%
}y_{\pm}\left(  k,x\right) \\
&  +\int R_{\pm}\left(  k\right)  e^{\pm2\mathrm{i}kx}y_{\pm}\left(
k,x\right)  \mathrm{d}k\\
&  =2\pi%
%TCIMACRO{\dsum _{n=1}^{N}}%
%BeginExpansion
{\displaystyle\sum_{n=1}^{N}}
%EndExpansion
c_{\pm,n}^{2}e^{\mp2\kappa_{n}x}y_{\pm}\left(  \mathrm{i}\kappa_{n},x\right)
+\int R_{\pm}\left(  k\right)  e^{\pm2\mathrm{i}kx}y_{\pm}\left(  k,x\right)
\mathrm{d}k
\end{align*}
and hence%
\[
q\left(  x\right)  =\pm\partial_{x}\left\{  2%
%TCIMACRO{\dsum _{n=1}^{N}}%
%BeginExpansion
{\displaystyle\sum_{n=1}^{N}}
%EndExpansion
c_{\pm,n}^{2}e^{\mp2\kappa_{n}x}\left[  1+y_{\pm}\left(  \mathrm{i}\kappa
_{n},x\right)  \right]  +\int R_{\pm}\left(  k\right)  e^{\pm2\mathrm{i}%
kx}\left[  1+y_{\pm}\left(  k,x\right)  \right]  \frac{\mathrm{d}k}{\pi
}\right\}  .
\]
Recalling that $y_{\pm}\left(  k,x\right)  =e^{\mp\mathrm{i}kx}\psi_{\pm
}(x,k)-1$ we finally obtain%
\[
q\left(  x\right)  =\pm\partial_{x}\left\{  2%
%TCIMACRO{\dsum _{n=1}^{N}}%
%BeginExpansion
{\displaystyle\sum_{n=1}^{N}}
%EndExpansion
c_{\pm,n}^{2}e^{\mp\kappa_{n}x}\psi_{\pm}(x,\mathrm{i}\kappa_{n})+\int
e^{\pm\mathrm{i}kx}R_{\pm}\left(  k\right)  \psi_{\pm}(x,k)\frac{\mathrm{d}%
k}{\pi}\right\}  ,
\]
which is (\ref{trace formula 3}).
\end{proof}

\begin{remark}
It follows from (\ref{asympt}) that $y_{\pm}\left(  k,x\right)  \notin L^{1}$
but $\operatorname{Re}y_{\pm}\left(  k,x\right)  \in L^{1}$.
\end{remark}

An important corollary of Theorem \ref{thm on trace} is the following

\begin{theorem}
\label{thm on trace formula 4}Suppose that $q\in L_{1}^{1}$. Let $\left\{
R,\kappa_{n},c_{n}\right\}  $ be its right scattering data and $\psi_{0}(x,k)$
be the right Jost solution corresponding to the data $\left\{  R,\varnothing
\right\}  $. Denote%
\[
\boldsymbol{\Psi}_{0}\left(  x\right)  =\left(  \psi_{0}\left(  x,\mathrm{i}%
\kappa_{n}\right)  \right)  ,\ \ \ \boldsymbol{C}=\operatorname*{diag}\left(
c_{n}^{2}\right)  .
\]
Then%
\begin{align}
q\left(  x\right)   &  =q_{0}\left(  x\right) \label{trace formula 4}\\
&  +2\partial_{x}\boldsymbol{\Psi}_{0}\left(  x\right)  \left(  \boldsymbol{C}%
^{-1}+\int_{x}^{\infty}\boldsymbol{\Psi}_{0}\left(  s\right)  ^{T}%
\boldsymbol{\Psi}_{0}\left(  s\right)  \mathrm{d}s\right)  ^{-1}%
\boldsymbol{\Psi}_{0}\left(  x\right)  ^{T},\nonumber
\end{align}
where $q_{0}\left(  x\right)  $ admits the following representations%
\begin{align*}
q_{0}\left(  x\right)   &  =2\partial_{x}\left\{  \int\operatorname{Re}\left[
1-e^{-\mathrm{i}kx}\psi_{0}(x,k)\right]  \frac{\mathrm{d}k}{\pi}\right\} \\
&  =\partial_{x}\left(  PV\right)  \int e^{\mathrm{i}kx}R\left(  k\right)
\psi_{0}(x,k)\frac{\mathrm{d}k}{\pi}\\
&  =\partial_{x}^{2}\int\frac{e^{2\mathrm{i}kx}-1}{2\mathrm{i}k}R\left(
k\right)  \frac{\mathrm{d}k}{\pi}+\partial_{x}\int e^{2\mathrm{i}kx}R\left(
k\right)  y_{0}\left(  k,x\right)  \frac{\mathrm{d}k}{\pi}.
\end{align*}

\end{theorem}

\begin{proof}
We merely combine the formula (\ref{trace formula2}) from Theorem
\ref{thm on trace} and the version of the binary Darboux transformation from
our \cite{RybSAM22}%
\[
q\left(  x\right)  =q_{0}\left(  x\right)  -2\partial_{x}^{2}\log\det\left(
\boldsymbol{C}^{-1}+\int_{x}^{\infty}\boldsymbol{\Psi}_{0}^{T}\left(
s\right)  \boldsymbol{\Psi}_{0}\left(  s\right)  \right)  ,
\]
where $q_{0}\left(  x\right)  $ is the potential corresponding to $\left\{
R,\varnothing\right\}  $. Indeed, by the Jacobi formula on differentiation of
determinants one has%
\begin{align*}
&  \partial_{x}\log\det\left(  \boldsymbol{C}^{-1}+\int_{x}^{\infty
}\boldsymbol{\Psi}_{0}^{T}\left(  s\right)  \boldsymbol{\Psi}_{0}\left(
s\right)  \right) \\
&  =-\boldsymbol{\Psi}_{0}\left(  x\right)  \left(  \boldsymbol{C}^{-1}%
+\int_{x}^{\infty}\boldsymbol{\Psi}_{0}\left(  s\right)  ^{T}\boldsymbol{\Psi
}_{0}\left(  s\right)  \mathrm{d}s\right)  ^{-1}\boldsymbol{\Psi}_{0}\left(
x\right)  ^{T}.
\end{align*}

\end{proof}

We will demonstrate below that the trace formula (\ref{trace formula 4}) is
convenient for limiting arguments. Of course, a similar formula holds for the
left scattering data.

\section{Trace formula and KdV solutions\label{sect of kdv solutions}}

In this section we show that our trace formulas yield new representations for
solutions to the KdV equation with short-range initial data. Note that the
condition $q\left(  x\right)  \in L_{1}^{1}$ alone does not guaranty that
$q\left(  x,t\right)  \in L_{1}^{1}$ for $t>0$ (see Corollary
\ref{Corol on preservation}) and therefore (\ref{trace formula2}) does not
apply. We cannot even be sure that (\ref{cond on y}) holds for $q\left(
x,t\right)  .$ To overcome the problems we employ some limiting arguments.
Through the rest of the paper we use the following convenient notation%
\[
\xi_{x,t}\left(  k\right)  :=\exp\mathrm{i}\left(  8k^{3}t+2kx\right)  .
\]
While highly oscillatory on the real line, this function has a rapid decay
along $\mathbb{R}+\mathrm{i}a$ for any $a>0$.

\begin{theorem}
\label{main thm}If $q\left(  x\right)  \in L_{1}^{1}$ and $\left\{
R,\kappa_{j},c_{j}\right\}  $ are the associated right scattering data then
the solution\footnote{The general theory guaranties well-posedness at least in
the $L^{2}$- base Sobolev space $H^{-\varepsilon}$ with any index
$0<\varepsilon\leq1$ (see, e.g. \cite{KillipWP2019}).} $q\left(  x,t\right)  $
to the Cauchy problem for the KdV equation (\ref{KdV}) with initial data
$q\left(  x\right)  $ can be represented by%
\begin{align}
q\left(  x,t\right)   &  =q_{0}\left(  x,t\right) \label{KdV solution}\\
&  +2\partial_{x}\boldsymbol{\Psi}_{0}\left(  x,t\right)  \left(
\boldsymbol{C}\left(  t\right)  ^{-1}+\int_{x}^{\infty}\boldsymbol{\Psi}%
_{0}\left(  s,t\right)  ^{T}\boldsymbol{\Psi}_{0}\left(  s,t\right)
\mathrm{d}s\right)  ^{-1}\boldsymbol{\Psi}_{0}\left(  x,t\right)
^{T},\nonumber
\end{align}
where%
\[
\boldsymbol{\Psi}_{0}\left(  x,t\right)  =\left(  \psi_{0}\left(
x,t,\mathrm{i}\kappa_{j}\right)  \right)  ,\ \ \ \boldsymbol{C}\left(
t\right)  =\left(  c_{j}\exp8\kappa_{j}^{3}t\right)  ,
\]%
\[
\psi_{0}\left(  x,t,k\right)  =e^{\mathrm{i}kx}\left[  1+y_{0}\left(
k,x,t\right)  \right]  ,
\]
$y_{0}\left(  \cdot,x,t\right)  $ is the $H^{2}$ solution of the singular
integral equation%
\[
y+\mathbb{H}(R\xi_{x,t})y=-\mathbb{H}(R\xi_{x,t})1,
\]%
\begin{equation}
q_{0}\left(  x,t\right)  =\partial_{x}\left\{  \left(  PV\right)  \int
R\left(  k\right)  \xi_{x,t}\left(  k\right)  \frac{\mathrm{d}k}{\pi}+\int
R\left(  k\right)  \xi_{x,t}\left(  k\right)  y_{0}\left(  k,x,t\right)
\frac{\mathrm{d}k}{\pi}\right\}  . \label{q0}%
\end{equation}

\end{theorem}

\begin{proof}
For Schwarz $q\left(  x\right)  $ there is nothing to prove as $q\left(
x,t\right)  $ is also a Schwarz function. Since KdV is well-posed in any
Sobolev space $H^{-\varepsilon}$ with $0<\varepsilon\leq1$ (see e.g.
\cite{KillipWP2019}) and $L_{1}^{1}\subset H^{-\varepsilon}$, for any sequence
of (real) Schwarz functions $q_{n}\left(  x\right)  $ approximating $q\left(
x\right)  $ in $L_{1}^{1}$ the sequence of $q_{n}\left(  x,t\right)  $
converges in $H^{-\varepsilon}$ to $q\left(  x,t\right)  $, the solution to
(\ref{KdV}) with the initial profile $q\left(  x\right)  $. Thus, we only need
to compute $\lim_{n\rightarrow\infty}q_{n}\left(  x,t\right)  $. Note that
convergence of norming constants is somewhat inconvenient to deal with but
results of our recent \cite{RybSAM22} offers a simple detour of this
circumstance. Take the scattering data $\left\{  R,\varnothing\right\}  $
(i.e. no bound states) and construct by (\ref{trace formula2}) the
corresponding potential%
\[
q_{0}\left(  x\right)  =-2\partial_{x}\int\operatorname{Re}y_{0}\left(
k,x\right)  \mathrm{d}k.
\]
Since by construction $\mathbb{L}_{q_{0}}$ is positive, $q_{0}$ is the Miura
transformation%
\[
q_{0}\left(  x\right)  =\partial_{x}r\left(  x\right)  +r\left(  x\right)
^{2}%
\]
of some real $r\in L_{\operatorname*{loc}}^{2}$ \cite{KapPerryTopalov2005}.
Choose a sequence $\left(  r_{n}\right)  $ of Schwarz function such that the
sequence $q_{0,n}=r_{n}^{2}+\partial_{x}r_{n}$ approximates $q_{0}$ in
$L_{1}^{1}$. As is well-known, each $R_{n}\left(  k\right)  $ is also Schwarz
and so is $q_{0,n}\left(  x,t\right)  $ for $t\geq0$. Therefore, by
(\ref{trace formula 3}) and recalling that $\psi_{n}=e^{\mathrm{i}kx}\left(
1+y_{n}\right)  $ we have%
\begin{align*}
q_{0,n}\left(  x,t\right)   &  =\partial_{x}\int R_{n}\left(  k\right)
\xi_{x,t}\left(  k\right)  \left[  1+y_{n}\left(  k,x,t\right)  \right]
\frac{\mathrm{d}k}{\pi}\\
&  =\partial_{x}^{2}\int\frac{\xi_{x,t}\left(  k\right)  -1}{2\mathrm{i}%
k}R_{n}\left(  k\right)  \frac{\mathrm{d}k}{\pi}+\partial_{x}\int R_{n}\left(
k\right)  \xi_{x,t}\left(  k\right)  y_{n}\left(  k,x,t\right)  \frac
{\mathrm{d}k}{\pi}\\
&  =:q_{0,n}^{\left(  1\right)  }\left(  x,t\right)  +q_{0,n}^{\left(
2\right)  }\left(  x,t\right)
\end{align*}
Here we have used a well-known regularization of the Fourier integral. This
representation is convenient for passing to the limit as $n\rightarrow\infty$.
By Proposition \ref{Prop on R} the sequence of reflection coefficients $R_{n}$
converges in $L^{2}$ to $R$. Paring this sequence, if needed, we may assume
that $R_{n}\rightarrow R$ a.e. Clearly $R_{n}\xi_{x,t}\rightarrow R\xi_{x,t}$
a.e. too. But then, as is well-known (can also be easily shown), the
corresponding sequence of Hankel operators $\mathbb{H}(R_{n}\xi_{x,t})$
converges to $\mathbb{H}(R\xi_{x,t})$ in the strong operator topology. Since
(all) $I+\mathbb{H}(R_{n}\xi_{x,t})$ and $I+\mathbb{H}(R\xi_{x,t})$ are
positive definite \cite[Theorem 8.2]{GruRybSIMA15} for all $x,t$ we conclude
that in $H^{2}$%
\begin{align*}
y_{n}  &  =-\left(  I+\mathbb{H}(R_{n}\xi_{x,t})\right)  ^{-1}\mathbb{H}%
(R_{n}\xi_{x,t})1\\
&  \rightarrow-\left(  I+\mathbb{H}(R\xi_{x,t})\right)  ^{-1}\mathbb{H}%
(R\xi_{x,t})1=:y_{0}\ \ \ n\rightarrow\infty,
\end{align*}
where $y_{0}\left(  \cdot,x,t\right)  \in H^{2}$. Therefore, for all $x,t$%
\[
\int\frac{\xi_{x,t}\left(  k\right)  -1}{2\mathrm{i}k}R_{n}\left(  k,t\right)
\frac{\mathrm{d}k}{\pi}\rightarrow\int\frac{\xi_{x,t}\left(  k\right)
-1}{2\mathrm{i}k}R\left(  k,t\right)  \frac{\mathrm{d}k}{\pi}%
\]
and%
\[
\int R\left(  k\right)  \xi_{x,t}\left(  k\right)  y_{n}\left(  k,x,t\right)
\frac{\mathrm{d}k}{\pi}\rightarrow\int R\left(  k\right)  \xi_{x,t}\left(
k\right)  y_{0}\left(  k,x,t\right)  \frac{\mathrm{d}k}{\pi}.
\]
Thus we conclude that for each $t\geq0$%
\[
w^{\ast}-\lim_{n\rightarrow\infty}q_{0,n}^{\left(  1\right)  }\left(
x,t\right)  =\partial_{x}^{2}\int\frac{\xi_{x,t}\left(  k\right)
-1}{2\mathrm{i}k}R\left(  k\right)  \frac{\mathrm{d}k}{\pi}%
\]%
\[
w^{\ast}-\lim_{n\rightarrow\infty}q_{0,n}^{\left(  2\right)  }\left(
x,t\right)  =\partial_{x}\int R_{n}\left(  k\right)  \xi_{x,t}\left(
k\right)  y_{n}\left(  k,x,t\right)  \frac{\mathrm{d}k}{\pi}%
\]
and (\ref{q0}) follows. Performing the binary Darboux transformation
\cite{RybSAM22} we arrive at (\ref{KdV solution}).
\end{proof}

\begin{remark}
Performing in (\ref{KdV solution}) the inverse binary Darboux transformation
\cite{RybSAM22}, we can conclude that we also have%
\begin{equation}
q\left(  x,t\right)  =-\frac{2}{\pi}\partial_{x}\int\operatorname{Re}y\left(
k,x,t\right)  \mathrm{d}k \label{q(x,t)}%
\end{equation}
but cannot claim that this integral is absolutely convergent as it was in
(\ref{trace formula2}). This of course would be true if the asymptotic
(\ref{asympt}) held for $q\left(  x,t\right)  $. The problem with
(\ref{asympt}) is that the error in (\ref{asympt}) depends on $\left\Vert
q\left(  \cdot,t\right)  \right\Vert _{L^{1}}$, which need not be finite. We
however conjecture that the integral in (\ref{q(x,t)}) is indeed absolutely
convergent but we can no longer use tools and estimates from the short-range
scattering theory.
\end{remark}

\begin{corollary}
The trace Deift-Trubowitz trace formula (\ref{DT trace}) holds for $t>0$ (not
only for $t=0$).
\end{corollary}

\begin{proof}
Indeed, the approximating sequence $q_{n}\left(  x,t\right)  $ that
corresponds to the sequence $\left\{  R_{n},\kappa_{j},c_{j}\right\}  $ where
$R_{n}$ is the same as constructed in the proof of Theorem \ref{main thm} will
do the job. The only question is why the first term in (\ref{DT trace}) holds
for $t>0$. This easily follows from our arguments. Indeed, since $y_{n}\left(
\cdot,x,t\right)  \rightarrow y\left(  \cdot,x,t\right)  $ in $H^{2}$ we also
have uniform convergence for $\operatorname{Im}k>0$ on compacts. Therefore,
$\psi_{n}\left(  \mathrm{i}\kappa_{j},x,t\right)  \rightarrow\psi\left(
\mathrm{i}\kappa_{j},x,t\right)  $ for all $x,t$.
\end{proof}

\begin{remark}
Extension of the Deift-Trubowitz trace formula (\ref{DT trace}) to KdV
solution would be a hard problem back in the 70s as the breakthrough in the
understanding of wellposedness in the $L^{2}$ based Sobolev spaces with
negative indexes only occurred after the seminal 1993 Bourgain paper
\cite{Bourgain93}, where wellposedness was proven in $L^{2}$. With no
well-posedness at hand we cannot use limiting arguments even if $q\left(
x\right)  \in L_{1}^{1}\cap L^{2}$.
\end{remark}

\section{How KdV trades decay for smoothness\label{sect on analytic smoothing}%
}

The goal of this section is to show how the results of the previous section
could be useful in understanding the phenomenon of dispersive smoothing (aka
gain of regularity).

\begin{theorem}
\label{Thm on decay}If $q\left(  x\right)  \in L_{1}^{1}\cap L^{2}$ then
$q\left(  x,t\right)  \in L_{\operatorname*{loc}}^{\infty}\cap L^{2}$ for
$t>0$.
\end{theorem}

\begin{proof}
We first note that we may assume that the negative spectrum is absent. I.e.
$\mathbb{L}_{q}$ is positive. Split our initial profile as%
\[
q=q_{-}+q_{+},\ \ \ q_{\pm}:=\left.  q\right\vert _{\mathbb{R}_{\pm}}.
\]
We may assume that $\mathbb{L}_{q_{+}}$ is positive as possible appearance of
a negative eigenvalue could only lead to minor technical complications. We use
the following representation from \cite{GruRybBLMS20}:%
\begin{equation}
R\left(  k\right)  =\phi_{1}\left(  k\right)  +\phi_{2}\left(  k\right)
+A\left(  k\right)  , \label{R-split}%
\end{equation}
where
\[
\phi_{1}\left(  k\right)  =\frac{T_{0}\left(  k\right)  }{2\mathrm{i}%
k}\widehat{q}\left(  k\right)  ,
\]%
\[
\phi_{2}\left(  k\right)  =\frac{T_{0}\left(  k\right)  }{\left(
2\mathrm{i}k\right)  ^{2}}\widehat{p}\left(  k\right)  ,
\]%
\[
\widehat{f}\left(  k\right)  :=\int_{0}^{\infty}e^{-2\mathrm{i}ks}p\left(
s\right)  \mathrm{d}s;
\]
$T_{0}\in H^{\infty}$ is the transmission coefficient for $q_{+}$; $p$ is the
derivative of an absolutely continuous function and%
\begin{equation}
\left\vert p\left(  x\right)  \right\vert \lesssim_{\left\Vert q_{+}%
\right\Vert _{L^{1}}}\left\vert q\left(  x\right)  \right\vert +C\int%
_{x}^{\infty}\left\vert q\right\vert ,\ \ x\geq0; \label{props of Q}%
\end{equation}
and $A\in H^{\infty}$ (which form is not important). Note that $\dfrac
{T_{0}\left(  k\right)  }{2\mathrm{i}k}$ remains bounded at $k=0$ as well as
$\dfrac{T_{0}\left(  k\right)  }{\left(  2\mathrm{i}k\right)  ^{2}}%
\widehat{p}\left(  k\right)  $.

It follows from (\ref{q0}) that%
\begin{align}
q_{0}\left(  x,t\right)   &  =\partial_{x}\left\{  \left(  PV\right)  \int
R\left(  k\right)  \xi_{x,t}\left(  k\right)  \frac{\mathrm{d}k}{\pi}+\int
R\left(  k\right)  \xi_{x,t}\left(  k\right)  y_{0}\left(  k,x,t\right)
\frac{\mathrm{d}k}{\pi}\right\} \nonumber\\
&  =\partial_{x}\left(  PV\right)  \int R\left(  k\right)  \xi_{x,t}\left(
k\right)  \frac{\mathrm{d}k}{\pi}\nonumber\\
&  +\int R\left(  k\right)  \left[  \partial_{x}\xi_{x,t}\left(  k\right)
\right]  y_{0}\left(  k,x,t\right)  \frac{\mathrm{d}k}{\pi}\nonumber\\
&  +\int R\left(  k\right)  \xi_{x,t}\left(  k\right)  \left[  \partial
_{x}y_{0}\left(  k,x,t\right)  \right]  \frac{\mathrm{d}k}{\pi}\nonumber\\
&  =:q_{1}\left(  x,t\right)  +q_{2}\left(  x,t\right)  +q_{3}\left(
x,t\right)  . \label{pieces}%
\end{align}
Consider each term separately. By (\ref{R-split}) for $q_{1}\left(
x,t\right)  $ we have%
\[
q_{1}\left(  x,t\right)  =q_{11}\left(  x,t\right)  +q_{12}\left(  x,t\right)
+q_{13}\left(  x,t\right)
\]
where%
\[
q_{1n}\left(  x,t\right)  :=\partial_{x}\left(  PV\right)  \int\xi
_{x,t}\left(  k\right)  \phi_{n}\left(  k\right)  \frac{\mathrm{d}k}{\pi
},\ \ \ n=1,2
\]
and%
\[
q_{13}\left(  x,t\right)  =\partial_{x}\left(  PV\right)  \int\xi_{x,t}\left(
k\right)  A\left(  k\right)  \frac{\mathrm{d}k}{\pi}.
\]
The simplest term is $q_{13}$. Since $\xi_{x,t}\left(  k+\mathrm{i}a\right)  $
(and all its $x$-derivatives) rapidly decays along $\mathbb{R+}\mathrm{i}a$
for any $a>0$ we deform the contour of integration to $\mathbb{R+}\mathrm{i}a$
that provides a rapid convergence of the integral (the original integral need
not be absolutely convergent).

The term $q_{12}$ is also easy. Indeed, since $\phi_{2}$ is clearly in $L^{1}$
we have%
\begin{align*}
q_{12}\left(  x,t\right)   &  =\partial_{x}\int\phi_{2}\left(  k\right)
\xi_{x,t}\left(  k\right)  \frac{\mathrm{d}k}{\pi}\\
&  =\int\frac{T_{0}\left(  k\right)  }{2\mathrm{i}k}\widehat{p}\left(
k\right)  \xi_{x,t}\left(  k\right)  \frac{\mathrm{d}k}{\pi}.
\end{align*}
It remains to show that this integral is absolutely convergent. It follows
from (\ref{props of Q}) that%
\begin{align*}
\left\Vert p\right\Vert  &  \lesssim_{\left\Vert q_{+}\right\Vert _{L^{1}}%
}\left\Vert q\right\Vert +C\left(  \int_{0}^{\infty}\left(  \int_{x}^{\infty
}\left\vert q\right\vert \right)  ^{2}\mathrm{d}x\right)  ^{1/2}\\
&  \leq\left\Vert q\right\Vert +C\left(  \int_{0}^{\infty}x\left\vert q\left(
x\right)  \right\vert \left(  \int_{x}^{\infty}\left\vert q\right\vert
\right)  \mathrm{d}x\right)  ^{1/2}\ \ \ \ \ \left(  q\in L_{1}^{1}\right) \\
&  \leq\left\Vert q\right\Vert +C\left\Vert q\right\Vert _{L^{1}}%
^{1/2}\left\Vert q\right\Vert _{L_{1}^{1}}^{1/2}<\infty
\end{align*}
and hence $\widehat{p}\in L^{2}$. Therefore, $q_{12}\left(  x,t\right)  $ is
locally bounded for $t\geq0$. (In fact, continuous).

Consider the remaining term $q_{11}\left(  x,t\right)  $. In order to proceed
we need first to regularize the improper integral. It cannot be done by merely
deforming $\mathbb{R}$ to $\mathbb{R+}\mathrm{i}a$ as is done for $q_{3}$
since $\widehat{q}\left(  k\right)  $ need not admit analytic continuation
into the upper half plane. To detour this circumstance we define
$\widehat{q}\left(  k\right)  $ by%
\[
\widehat{q}\left(  \overline{k}\right)  =\int_{0}^{\infty}e^{-2\mathrm{i}%
\overline{k}s}q\left(  s\right)  \mathrm{d}s,\ \ \ \operatorname{Im}k\geq0,
\]
and apply the Cauchy-Green formula for the strip $0\leq\operatorname{Im}%
k\leq1$. We have%
\begin{align*}
q_{11}\left(  x,t\right)   &  =\partial_{x}\left(  PV\right)  \int\xi
_{x,t}\left(  k\right)  \frac{T_{0}\left(  k\right)  }{2\mathrm{i}%
k}\widehat{q}\left(  \overline{k}\right)  \frac{\mathrm{d}k}{\pi}\\
&  =\partial_{x}\int_{0\leq\operatorname{Im}k\leq1}\xi_{x,t}\left(  k\right)
\frac{T_{0}\left(  k\right)  }{2\mathrm{i}k}\partial_{\overline{k}}%
\widehat{q}\left(  \overline{k}\right)  \frac{\mathrm{d}u\mathrm{d}v}{\pi^{2}%
}\text{ \ \ }(k=u+\mathrm{i}v)\\
&  +\partial_{x}\int_{\mathbb{R+}\mathrm{i}}\xi_{x,t}\left(  k\right)
\frac{T_{0}\left(  k\right)  }{2\mathrm{i}k}\widehat{q}\left(  \overline
{k}\right)  \frac{\mathrm{d}k}{\pi}\\
&  =:q_{111}\left(  x,t\right)  +q_{112}\left(  x,t\right)  .
\end{align*}
The second term $q_{112}\left(  x,t\right)  $ is treated the same way as
$q_{13}\left(  x,t\right)  $ and one immediately concludes that $q_{112}%
\left(  x,t\right)  $ are bounded (in fact, smooth) for $t>0$. Turn to
$q_{111}$. We rearrange it by observing that the double integral is absolutely
convergent and the order of integration may be interchanged:
\begin{align*}
q_{111}\left(  x,t\right)   &  =\partial_{x}\int_{0\leq\operatorname{Im}%
k\leq1}\xi_{x,t}\left(  k\right)  \frac{T_{0}\left(  k\right)  }{2\mathrm{i}%
k}\left[  \int_{0}^{\infty}\left(  -2\mathrm{i}s\right)  e^{-2\mathrm{i}%
\overline{k}s}q\left(  s\right)  \mathrm{d}s\right]  \frac{\mathrm{d}%
u\mathrm{d}v}{\pi^{2}}\\
&  =-\partial_{x}\int_{0}^{\infty}sq\left(  s\right)  \left[  \int%
_{0\leq\operatorname{Im}k\leq1}\xi_{x,t}\left(  k\right)  \frac{T_{0}\left(
k\right)  }{k}e^{-2\mathrm{i}\overline{k}s}\frac{\mathrm{d}u\mathrm{d}v}%
{\pi^{2}}\right]  \mathrm{d}s\\
&  =-\mathrm{i}\int_{0}^{\infty}\left[  \int_{0\leq\operatorname{Im}k\leq
1}e^{-2vs}\xi_{x-s,t}\left(  k\right)  T_{0}\left(  k\right)  \frac
{\mathrm{d}u\mathrm{d}v}{\pi^{2}}\right]  sq\left(  s\right)  \mathrm{d}%
s,\ \ \ (\overline{k}=k-2\mathrm{i}v)\\
&  \simeq\int_{0}^{\infty}\left\{  \int_{0}^{1}e^{-2vs}I\left(  s-x,t\right)
\mathrm{d}v\right\}  sq\left(  s\right)  \mathrm{d}s,
\end{align*}
where%
\begin{align*}
I\left(  s-x,t\right)   &  :=\int_{\mathbb{R}+\mathrm{i}v}\xi_{x-s,t}\left(
k\right)  T_{0}\left(  k\right)  \mathrm{d}k\\
&  =\int_{\mathbb{R}+\mathrm{i}}\xi_{x-s,t}\left(  k\right)  T_{0}\left(
k\right)  \mathrm{d}k
\end{align*}
is independent of $v\geq0$. Thus%
\begin{align*}
q_{111}\left(  x,t\right)   &  \simeq\partial_{x}\int_{0}^{\infty}\left\{
\int_{0}^{1}e^{-2vs}\mathrm{d}v\right\}  I\left(  s-x,t\right)  sq\left(
s\right)  \mathrm{d}s\\
&  =\int_{0}^{\infty}\frac{1-e^{-2s}}{2}I\left(  s-x,t\right)  q\left(
s\right)  \mathrm{d}s\\
&  =\int_{0}^{\infty}I\left(  s-x,t\right)  \frac{1-e^{-2s}}{2}q\left(
s\right)  \mathrm{d}s.
\end{align*}
It remains to study the behavior of $I\left(  s-x,t\right)  $ as
$s\rightarrow+\infty$. Since $T_{0}\left(  k\right)  =1+O\left(
k^{-1}\right)  $ as $k\rightarrow\infty$ and $x$ is fixed we only need to
worry about%
\[
I_{0}\left(  s,t\right)  =\int_{\mathbb{R}+\mathrm{i}}\xi_{-s,t}\left(
k\right)  \mathrm{d}k,
\]
which is closely related to the Airy function. For the reader convenience we
offer a direct treatment. Rewrite%
\[
\xi_{-s,t}\left(  k\right)  =\exp\mathrm{i}\left[  \Omega S\left(
\lambda\right)  \right]  ,
\]
where $S\left(  \lambda\right)  =\lambda^{3}/3-\lambda$ and%
\[
\omega=2\left(  s/3t\right)  ^{1/2},\Omega=3t\left(  s/3t\right)
^{3/2},\lambda=k/\omega
\]
Noticing that we need not adjust the contour of integration, we then have%
\begin{equation}
I_{0}\left(  s,t\right)  :=\omega\int_{\mathbb{R}+\mathrm{i}}e^{\mathrm{i}%
\Omega S\left(  \lambda\right)  }\mathrm{d}\lambda. \label{IF}%
\end{equation}
Apparently, the phase $S\left(  \lambda\right)  =$ $\lambda^{3}/3-\lambda$ has
stationary points at $\lambda=\pm1$ and we need to deform the contour in
(\ref{IF}) to pass through points $\lambda=\pm1$. We denote such a contour
$\Gamma$. To apply the steepest descent we need to make sure that
$\exp\mathrm{i}\left[  \omega S\left(  \lambda\right)  \right]  $ decay on
$\Gamma$ away from $\pm1$. To this end $\Gamma$ must be in the lower half
plane between points $-1$ and $1$. Noticing that $\Omega=\left(  3t/8\right)
\omega^{3}$, $\omega=O\left(  s^{1/2}\right)  $ by the steepest descent method
(see e.g. \cite{Wong2001}) one has%
\begin{align*}
I_{0}\left(  s,t\right)   &  =\omega\int_{\Gamma}e^{\mathrm{i}\Omega S\left(
\lambda\right)  }\mathrm{d}\lambda=\omega O\left(  \Omega^{-1/2}\right)
,\ \ \ \Omega\rightarrow+\infty,\\
&  =O\left(  s^{-1/4}\right)  ,\ \ \ s\rightarrow+\infty.
\end{align*}
Thus $q_{111}\left(  x,t\right)  $ is bounded for $t>0$ (even if $q\left(
x\right)  $ decays slower than $L^{1}$). All four pieces $q_{1}\left(
x,t\right)  $ is made of are bounded and so is $q_{1}\left(  x,t\right)  $.

There is now only one term $q_{3}$ left in (\ref{pieces}) to analyze. We are
done if we show that $\partial_{x}y_{0}\in H^{2}$. Differentiating%
\[
y_{0}+\mathbb{H}y_{0}=-\mathbb{H}1,\ \ \ \mathbb{H}:=\mathbb{H}\left(
R\xi_{x,t}\right)  ,
\]
in $x$ one has%
\[
\partial_{x}y_{0}+\mathbb{H}\left(  \partial_{x}y_{0}\right)  =-\partial
_{x}\mathbb{H}1-\left(  \partial_{x}\mathbb{H}\right)  y_{0}.
\]
Thus%
\[
\partial_{x}y_{0}=-\left(  I+\mathbb{H}\right)  ^{-1}\left[  \left(
\partial_{x}\mathbb{H}\right)  1+\left(  \partial_{x}\mathbb{H}\right)
y_{0}\right]  .
\]
It follows that we only need to show that $\left(  \partial_{x}\mathbb{H}%
\right)  1\in H^{2}$ and $\partial_{x}\mathbb{H}$ is a bounded operator. Note
first that%
\[
\left(  \partial_{x}\mathbb{H}\right)  =\mathbb{H}\left(  2\mathrm{i}%
kR\xi_{x,t}\right)  .
\]
Since $kR\left(  k\right)  \in L^{2}$ (from the second Zakharov-Faddeev trace
formula),%
\[
\left(  \partial_{x}\mathbb{H}\right)  1=\mathbb{JP}_{-}\left(  2\mathrm{i}%
kR\left(  k\right)  \xi_{x,t}\left(  k\right)  \right)  \in H^{2}%
\]
as desired. The proof of boundedness of $\left(  \partial_{x}\mathbb{H}%
\right)  $ is a bit more complicated. By (\ref{R-split}) we have%
\[
\mathbb{H=H}_{1}+\mathbb{H}_{2}+\mathbb{H}_{3}%
\]
where%
\[
\mathbb{H}_{n}:=\mathbb{H}\left(  \phi_{n}\xi_{x,t}\right)
,n=1,2,\ \ \ \mathbb{H}_{3}:=\mathbb{H}\left(  A\xi_{x,t}\right)  .
\]
For $n=1,2$ both $\mathbb{H}_{n}$ admit a direct differentiation in $x$.
Indeed, one can easily see that%
\[
\partial_{x}\mathbb{H}_{n}=\mathbb{H}\left(  \phi_{n}\partial_{x}\xi
_{x,t}\right)  =\mathbb{H}\left(  2\mathrm{i}k\phi_{n}\xi_{x,t}\right)
,n=1,2.
\]
Since $q,p\in L^{1}$%
\begin{equation}
2\mathrm{i}k\phi_{1}\left(  k\right)  =T_{0}\left(  k\right)  \widehat{q}%
\left(  k\right)  \in L^{\infty} \label{R-split 1}%
\end{equation}
and%
\[
2\mathrm{i}k\phi_{2}\left(  k\right)  =\frac{T_{0}\left(  k\right)
}{2\mathrm{i}k}\widehat{p}\left(  k\right)  \in L^{\infty}%
\]
and hence the operators $\partial_{x}\mathbb{H}\left(  \phi_{n}\right)
,n=1,2,$ are bounded. To differentiate $\mathbb{H}_{3}$ we need first to use
(\ref{regularized Hankel}). One has%
\[
\mathbb{H=H}\left(  \mathbb{\bar{P}}_{-}\left(  R\xi_{x,t}\right)  \right)  .
\]
But
\begin{align*}
\mathbb{P}_{-}\left[  A\left(  k\right)  \xi_{x,t}\left(  k\right)  \right]
&  =-\frac{1}{2\mathrm{i}\pi}\int\frac{A\left(  \lambda\right)  \xi
_{x,t}\left(  \lambda\right)  }{\lambda-\left(  k-\mathrm{i}0\right)
}\mathrm{d}\lambda\\
&  =-\frac{1}{2\mathrm{i}\pi}\int_{\mathbb{R}+\mathrm{i}}\frac{A\left(
\lambda\right)  \xi_{x,t}\left(  \lambda\right)  }{\lambda-k}\mathrm{d}%
\lambda,
\end{align*}
where the integral is absolutely convergent, and therefore we may
differentiate under the integral sign%
\[
\partial_{x}\mathbb{P}_{-}\left[  A\left(  k\right)  \xi_{x,t}\left(
k\right)  \right]  =-\frac{1}{2\mathrm{i}\pi}\int_{\mathbb{R}+\mathrm{i}}%
\frac{2\mathrm{i}\lambda A\left(  \lambda\right)  \xi_{x,t}\left(
\lambda\right)  }{\lambda-k}\mathrm{d}\lambda,
\]
which is well-defined and bounded. Consequently, $\partial_{x}\mathbb{H}_{3}$
is a bounded operator and so is $\partial_{x}\mathbb{H}$. Thus, indeed
$\partial_{x}y_{0}\in H^{2}$.
\end{proof}

\begin{remark}
\label{remark in last section}Theorem 7.1 of \cite{GruRyb2018}, which proof is
based on the Dyson formula, relates smoothness of $q\left(  x,t\right)  $ with
the decay of $q\left(  x\right)  $. In particular, it follows from that result
that if $q\left(  x\right)  \in L_{3/2}^{1}\cap L^{2}$ then $q\left(
x,t\right)  \in L_{\operatorname*{loc}}^{\infty}\cap L^{2}$ for $t>0.$
Stronger decay is due to the fact the Dyson formula involves $\det\left(
I+\mathbb{H}\right)  $, which use requires to analyze differentiability of
$\mathbb{H}$ in trace norm. (The latter is also technically much more
involved. It was our attempt to dispose of trace norm considerations that led
us to our trace formulas, which require uniform norms only.
\end{remark}

The following important consequence directly follows from Theorem
\ref{Thm on decay} and invariance of the KdV\ with respect to $\left(
x,t\right)  \rightarrow\left(  -x,-t\right)  $.

\begin{corollary}
\label{Corol on preservation}The class $L_{1}^{1}$ is not preserved under the
KdV flow.
\end{corollary}

\begin{proof}
Suppose to the contrary that $L_{1}^{1}$ is preserved under the KdV flow. I.e.
if $q\left(  x\right)  \in L_{1}^{1}$ then $q\left(  x,t\right)  \in L_{1}%
^{1}$ for any $t$. Take $q\left(  x\right)  \in L_{1}^{1}\cap L^{2}$ but
$q\left(  x\right)  \notin L^{\infty}$ and fix $t_{0}>0$ \ By Theorem
\ref{Thm on decay}, $q_{0}\left(  x\right)  :=q\left(  x,t_{0}\right)  \in
L_{\operatorname*{loc}}^{\infty}\cap L^{2}$. Take $q_{0}\left(  x\right)  $ as
new initial data. By our assumption it is also in $L_{1}^{1}$. Thus
$q_{0}\left(  x\right)  \in L_{\operatorname*{loc}}^{\infty}\cap L_{1}^{1}\cap
L^{2}$. But this leads us to a contradiction as $q_{0}\left(  x,t_{0}\right)
=q\left(  -x\right)  $ was not assumed locally bounded.
\end{proof}

In the conclusion we mention that much more general and precise statements can
be made regarding how the KdV solutions gain regularity (smoothness) in
exchange for loss of decay. We plan on showing elsewhere how the results of
\cite{GruRyb2018}, \cite{GruRybBLMS20}, \cite{RybCommPDEs2013}, and
\cite{RybOTAA23} may be improved to optimal statements.

\section{Appendix\label{appendix}}

We demonstrate that the Deift-Trubowitz trace formula is actually a
"nonlinearization" of our trace formulas. Assume for simplicity that there are
no bound states (non-empty negative spectrum merely complicates the
computations) and do our computation for the $+$ sign only. The reader who has
been able to get to this point should be able to follow the calculations
below. Denoting $\mathbb{H=H}\left(  R\xi_{x,t}\right)  $, $h:=\mathbb{H}1$,
$1_{a}:=\chi_{\left\vert \cdot\right\vert \leq a}$, we have%

\begin{align*}
\pi q  &  =-2\partial_{x}\int\operatorname{Re}y=\frac{1}{\pi}\partial_{x}%
\int\operatorname{Re}\left(  I+\mathbb{H}\right)  ^{-1}h\\
&  =2\partial_{x}\int\operatorname{Re}\left(  I+\mathbb{H}\right)  ^{-1}h\\
&  =-2\int\left(  I+\mathbb{H}\right)  ^{-1}\left(  \partial_{x}%
\mathbb{H}\right)  \mathbb{\left(  I+\mathbb{H}\right)  }^{-1}h+2\int\left(
I+\mathbb{H}\right)  ^{-1}\partial_{x}h\\
&  =:q_{1}+q_{2}.
\end{align*}
For $q_{1}$ we have%
\[
q_{1}=-2\lim_{a\rightarrow\infty}\left\langle \left(  I+\mathbb{H}\right)
^{-1}\left(  \partial_{x}\mathbb{H}\right)  \mathbb{\left(  I+\mathbb{H}%
\right)  }^{-1}h,1_{a}\right\rangle .
\]
For the inner product one has%
\begin{align*}
&  \left\langle \left(  I+\mathbb{H}\right)  ^{-1}\left(  \partial
_{x}\mathbb{H}\right)  \mathbb{\left(  I+\mathbb{H}\right)  }^{-1}%
h,\mathbb{P}_{+}1_{a}\right\rangle \\
&  =\left\langle \left(  \partial_{x}\mathbb{H}\right)  \mathbb{\left(
I+\mathbb{H}\right)  }^{-1}h,\left(  I+\mathbb{H}\right)  ^{-1}\mathbb{P}%
_{+}1_{a}\right\rangle \\
&  =-\left\langle \left(  \partial_{x}\mathbb{H}\right)  y,\left(
I+\mathbb{H}\right)  ^{-1}\mathbb{P}_{+}1_{a}\right\rangle \\
&  =-\left\langle \left(  \partial_{x}\mathbb{H}\right)  y,\mathbb{P}%
_{+}\left[  1_{a}-\left(  I+\mathbb{H}\right)  ^{-1}\mathbb{HP}_{+}%
1_{a}\right]  \right\rangle \\
&  =-\left\langle \left(  \partial_{x}\mathbb{H}\right)  y,\mathbb{P}_{+}%
1_{a}\right\rangle +\left\langle \left(  \partial_{x}\mathbb{H}\right)
,\left(  I+\mathbb{H}\right)  ^{-1}\mathbb{HP}_{+}1_{a}\right\rangle .
\end{align*}
Passing to the limit yields%
\[
q_{1}=2\int\left(  \partial_{x}\mathbb{H}\right)  y+\left\langle \left(
\partial_{x}\mathbb{H}\right)  y,y\right\rangle .
\]
One may now see how "nonlinear" dependence on $y$ in (\ref{DT trace}) comes
about. Indeed, the second term $\left\langle \left(  \partial_{x}%
\mathbb{H}\right)  y,y\right\rangle $ is a quadratic form. For $q_{2}$ we
similarly have%
\begin{align*}
q_{2}  &  =2\int\left(  I+\mathbb{H}\right)  ^{-1}\partial_{x}h\\
&  =2\lim_{a\rightarrow\infty}\left\langle \left(  I+\mathbb{H}\right)
^{-1}\left(  \partial_{x}h\right)  ,\mathbb{P}_{+}1_{a}\right\rangle \\
&  =2\lim_{a\rightarrow\infty}\left\langle \partial_{x}h,\mathbb{P}_{+}%
1_{a}-\left(  I+\mathbb{H}\right)  ^{-1}\mathbb{HP}_{+}1_{a}\right\rangle \\
&  =\int\partial_{x}h+\left\langle \partial_{x}h,y\right\rangle .
\end{align*}
Since%
\[
\partial_{x}\mathbb{H}f\left(  k\right)  \mathbb{=}2\mathrm{i}\mathbb{JP}%
_{-}\left[  kR\left(  k\right)  e^{2\mathrm{i}kx}f\left(  k\right)  \right]
\]
we have%
\begin{align*}
\pi q  &  =q_{1}+q_{2}\\
&  =\left\langle \left(  \partial_{x}\mathbb{H}\right)  y,y\right\rangle
+2\int\left(  \partial_{x}\mathbb{H}\right)  y+\left\langle \partial
_{x}h,y\right\rangle +2\int\partial_{x}h\\
&  =2\mathrm{i}\int kR\left(  k\right)  e^{2\mathrm{i}kx}y\left(  k,x\right)
^{2}\mathrm{d}k\\
&  +4\mathrm{i}\int kR\left(  k\right)  e^{2\mathrm{i}kx}y\left(  k,x\right)
\mathrm{d}k+2\mathrm{i}\int kR\left(  k\right)  e^{2\mathrm{i}kx}\mathrm{d}k\\
&  =2\mathrm{i}\int kR\left(  k\right)  e^{2\mathrm{i}kx}\left[  1+y\left(
k,x\right)  \right]  ^{2}\mathrm{d}k\\
&  =2\mathrm{i}\int kR\left(  k\right)  \psi\left(  x,k\right)  ^{2}%
\mathrm{d}k
\end{align*}
and (\ref{DT trace}) with $c_{n}=0$ follows.


\begin{thebibliography}{99}                                                                                               %


\bibitem {BinderetalDuke18}Binder, Ilia; Damanik, David; Goldstein, Michael;
Lukic, Milivoje Almost periodicity in time of solutions of the KdV equation.
Duke Math. J. 167 (2018), no. 14, 2633--2678.

\bibitem {Blower23}Blower, Gordon; Doust, Ian Linear systems, Hankel products
and the sinh-Gordon equation. J. Math. Anal. Appl. 525 (2023), no. 1, Paper
No. 127140, 24 pp.

\bibitem {Bourgain93}Bourgain, J. Fourier transform restriction phenomena for
certain lattice subsets and applications to nonlinear evolution equations. II.
The KdV-equation. Geom. Funct. Anal. 3 (1993), no. 3, 209--262.

\bibitem {Craig89}Craig, W. The trace formula for Schr\"{o}dinger operators on
the line, Commun. Math. Phys, 126 (1989), 379-407.

\bibitem {Deift79}Deift, P.; Trubowitz, E. Inverse scattering on the
line\emph{.} Comm. Pure Appl. Math. 32 (1979), no. 2, 121--251.

\bibitem {Doikouetal21}Doikou, Anastasia; Malham, Simon J. A.; Stylianidis,
Ioannis Grassmannian flows and applications to non-commutative non-local and
local integrable systems. Phys. D 415 (2021), Paper No. 132744, 13 pp.

\bibitem {GGKM67}Gardner, C. S.; Greene, J. M.; Kruskal, M. D.; and Miura, R.
M. Method for Solving the Korteweg-deVries Equation\emph{.} Phys. Rev. Lett.
19 (1967), 1095--1097.

\bibitem {GerardPushnitski2023}G\'{e}rard, Patrick; Pushnitski, Alexander
Unbounded Hankel operators and the flow of the cubic Szeg\H{o} equation.
Invent. Math. 232 (2023), no. 3, 995--1026.

\bibitem {GesztesyBOOK08}Gesztesy, Fritz; Holden, Helge; Michor, Johanna;
Teschl, Gerald Soliton equations and their algebro-geometric solutions. Vol.
I. (1+1)-dimensional continuous models, Cambridge Stud. Adv. Math., 79,
Cambridge University Press, Cambridge, 2003. x+438 pp.

\bibitem {Gesztesyetal93}Gesztesy, F.; Holden, H.; Simon, B;.and Zhao, Z.
Trace formulae and inverse spectral theory for Schr\"{o}dinger operators,
Bull. Am. Math. Soc. 29 (1993), 250-255.

\bibitem {Gesztesyet95}Gesztesy, F.; Holden, H.; Simon, B; and Zhao, Z. Higher
order trace relations for Schr\"{o}dinger operators, Rev. Math. Phys. (1995), 893--922.

\bibitem {GesztesyHolden95}Gesztesy, F.; Holden, H. On new trace formulae for
Schr\"{o}dinger operators, Acta Applicandae Math. 39 (1995), 315--333.

\bibitem {GesztesySimon96}Gesztesy, F. and Simon, B. Uniqueness theorems in
inverse spectral theory for one-dimensional Schr\"{o}dinger operators, Trans.
Amer. Math. Soc. 348 (1996), 349--373.

\bibitem {GesztesyXi96}Gesztesy, F. and Simon, B. The xi function, Acta Math.
176 (1996), 49--71.

\bibitem {Gruetal23}Grudsky, Sergei M.; Kravchenko, Vladislav V.; Torba,
Sergii M. Realization of the inverse scattering transform method for the
Korteweg--de Vries equation. Math. Methods Appl. Sci. 46 (2023), no. 8, 9217--9251.

\bibitem {GruRybSIMA15}Grudsky, S.; Rybkin, A. Soliton theory and Hakel
operators, SIAM J. Math. Anal. 47 (2015), no. 3, 2283--2323.

\bibitem {GruRyb2018}Grudski\u{\i}, S. M.; Rybkin, A. V. On the trace-class
property of Hankel operators arising in the theory of the Korteweg--de Vries
equation. (Russian) Mat. Zametki 104 (2018), no. 3, 374--395; translation in
Math. Notes 104 (2018), no. 3-4, 377--394

\bibitem {GruRybBLMS20}Grudsky S, Rybkin A. On classical solution to the KdV
equation. Proc Lond Math Soc. 2020; 121(3): 354--371.

\bibitem {Hryniv21}Hryniv, Rostyslav; Mykytyuk, Yaroslav On the first trace
formula for Schr\"{o}dinger operators. J. Spectr. Theory 11 (2021), no. 2, 489--507.

\bibitem {KapPerryTopalov2005}Kappeler, T.; Perry, P.; Shubin, M.; Topalov, P.
The Miura map on the line. Int. Math. Res. Not. (2005), no. 50, 3091--133.

\bibitem {KillipSimon2008}Killip, Rowan; Simon, Barry Sum rules and spectral
measures of Schr\"{o}dinger operators with $L^{2}$ potentials. Ann. of Math.
(2) 170 (2009), no. 2, 739--782.

\bibitem {KillipWP2019}Killip, R.; Visan, M. KdV is wellposed in\emph{
}$H^{-1}$ Ann. of Math. (2) 190 (2019), no. 1, 249--305.

\bibitem {Malham22}Malham, Simon J. A. The non-commutative Korteweg--de Vries
hierarchy and combinatorial P\"{o}ppe algebra. Phys. D 434 (2022), Paper No.
133228, 25 pp.

\bibitem {MarchBook2011}Marchenko, Vladimir A. Sturm-Liouville operators and
applications. Revised edition. AMS Chelsea Publishing, Providence, RI, 2011.
xiv+396 pp.

\bibitem {McKean75}McKean, H. P.; Moerbeke, P. The spectrum of Hill's
equation, Invent. Math. 30 (1975), 217-274.

\bibitem {Nik2002}Nikolski, N. K. Operators, functions, and systems: An easy
reading. Volume 1: Hardy, Hankel and Toeplitz\emph{.} Mathematical Surveys and
Monographs, vol. 92, Amer. Math. Soc., Providence, 2002. 461 pp.

\bibitem {RemlingCMP15}Remling, Christian Generalized reflection coefficients.
Comm. Math. Phys. 337 (2015), no. 2, 1011--1026.

\bibitem {RybSIAM2001}Rybkin, Alexei On the trace approach to the inverse
scattering problem in dimension one. SIAM J. Math. Anal. 32 (2001), no. 6, 1248--1264.

\bibitem {RybCommPDEs2013}Rybkin, Alexei. Spatial analyticity of solutions to
integrable systems. I. The KdVcase\emph{ }Communications in Partial
Differential Equations, Volume 38 (2013), Issue 5, 802-822.

\bibitem {RybSAM22}Rybkin, Alexei The binary Darboux transformation revisited
and KdV solitons on arbitrary short-range backgrounds. Stud. Appl. Math. 148
(2022), no. 1, 141--153.

\bibitem {RybOTAA23}Rybkin, Alexei A trace formula and classical solutions to
the KdV equation. Oper. Theory Adv. Appl., 291 (2023), 667--677.

\bibitem {Trubowitz77}Trubowitz, E. The inverse problem for periodic
potentials, Comm. Pure Appl. Math. 30 (1977), 321-337.

\bibitem {Wong2001}Wong, R.S.C. \emph{A}symptotic Approximations of
Integrals.\emph{ }Academic Press, Inc., 2001. xiv + 540 pp. ISBN: 978-0-89871-497-5.

\bibitem {Zakharov71}Zaharov, V. E.; Faddeev, L. D. The Korteweg-de Vries
equation is a fully integrable Hamiltonian system. (Russian) Funkcional. Anal.
i Prilo\v{z}en. 5 (1971), no. 4, 18--27.
\end{thebibliography}
\end{document}